\documentclass[12pt]{article}

\usepackage[english]{babel}
\usepackage{amsfonts}
\usepackage{epsfig}
\usepackage{color}
\usepackage{amsmath}
\usepackage{amsthm}
\usepackage{amssymb}
\usepackage{tikz-cd}

\usepackage[normalem]{ulem}

\usepackage{comment}

\newcommand{\as}{\widetilde P^*(a,b)}
\newcommand{\bs}{\widetilde Q^*(a,b)}
\newcommand{\ps}{ P^*(a,b)}
\newcommand{\qs}{ Q^*(a,b)}
\newcommand{\pqs}{ q P^*(a,b)+ p Q^*(a,b)}

\newcommand{\aaa}{{\cal A}}

\newcommand{\ba}{\begin{array}}
\newcommand{\bal}{\begin{array}{l}}
\newcommand{\be}{\begin{equation}}
\newcommand{\beqa}{\begin{eqnarray}}
\newcommand{\bl}{\begin{lem}}
\newcommand{\bt}{\begin{teo}}

\newcommand{\C}{\mathbb{C}}
\newcommand{\ch}{\choose}

\newcommand{\ea}{\end{array}}
\newcommand{\ee}{\end{equation}}
\newcommand{\eeqa}{\end{eqnarray}}
\newcommand{\el}{\end{lem}}
\newcommand{\et}{\end{teo}}

\newcommand{\gk}{g_{qk,pk} }

\newcommand{\kxn}{k[x_1,\dots,x_n]}

\newcommand{\la}{\langle}

\newcommand{\mc}{\mathcal}

\newcommand{\N}{\mathbb{N}}

\newcommand{\Q}{\mathbb{C}}

\newcommand{\ra}{\rangle}

\newcommand{\vv}{{\bf V}}

\newcommand{\wt}{\widetilde}

\newcommand{\Z}{\mathbb{Z}}

\newcommand{\sH}{{\sum_{(i,j)\in {S_H}}}}

\newcommand{\monom}{[\mu;\nu]}

\newtheorem{teo}{Theorem}[section]
\newtheorem{lem}[teo]{Lemma}
\newtheorem{pro}[teo]{Proposition} 
\newtheorem{defin}[teo]{Definition}
\newtheorem{prob}{Problem}        
\newtheorem*{exstar}{Example}             
\newtheorem{conj}{Conjecture}

\theoremstyle{plain}
\newtheorem{remark}[teo]{Remark}

\providecommand{\keywords}[1]
{
  \text{\textit{Keywords---}} #1
}

\title{Integrability of polynomial vector fields \\ and a dual problem}
\author{Tatjana Petek$^{1,2}$  and Valery G.~Romanovski$^{1,3,4}$\\
 $^1${\it Faculty of Electrical Engineering and Computer Science,} \\ {\it University of Maribor,
 Koro\v ska cesta 46, SI-2000 Maribor, Slovenia} \\
$^2${\it Institute of Mathematics, Physics and Mechanics,}\\
{\it Jadranska 19, SI-1000 Ljubljana, Slovenia}\\
$^3${\it Faculty of Natural Science and Mathematics,}   \\ {\it University of Maribor,
 Koro\v ska cesta 160, SI-2000 Maribor, Slovenia}\\
$^4${\it Center for Applied Mathematics and Theoretical Physics,}\\
{\it Mladinska 3, SI-2000 Maribor, Slovenia}}
\date{}

\begin{document}
 \maketitle  

 \begin{abstract} 
 We  investigate the integrability of polynomial vector fields through the lens of duality in parameter spaces. We examine formal power series solutions annihilated by  differential operators and explore the properties of the integrability variety in relation to the invariants of the associated Lie group. Our study extends to differential operators on affine algebraic varieties, highlighting the intrinsic  connection between these operators and local analytic first integrals. To illustrate the duality  the case of quadratic vector fields is considered in  detail. 
 
\end{abstract}


\section{Introduction}


This paper delves into one of the  pivotal aspects of polynomial systems of ODEs -- integrability.  Integrability pertains to the ability to solve a system exactly, yielding explicit solutions, or to find  conserved quantities that simplify the system's analysis. Identifying integrable systems not only enriches our theoretical understanding but also provides powerful tools for performing the bifurcational analysis of the systems.

In this paper, we deal with the  integrability of
two-dimensional polynomial systems 
of the form 
\be  \label{gsG}
\begin{aligned}
\frac{d x}{d\tau} &= \phantom{-}  p  x - \sum_{i+j=2}^m
                                   \alpha_{ij}x^{i}{y}^{j}= X(x,y), \\
\frac{d y}{d\tau} &=             -q y+\sum_{i+j=2}^m
                                  \beta_{ij}x^{i}{y}^{j}=Y(x,y),
\end{aligned}
\ee 
where $\alpha_{ij}$ and $\beta_{ij}$ are complex parameters 
and $p, q$ mutually prime natural numbers. 
We  first address the problem of finding systems 
in a given family \eqref{gsG} 
which admit an analytic first integral in a neighborhood of the origin. The problem is known as the problem of local analytic  integrability or, in the case $p=q=1$, as the Poincar\'e center problem,   which has been intensively studied for more than a century {(see, e.g. \cite{ChR,FSZ,GM,JLR,Liu2,Huang,RS,Sib1,YZ,Z}} and the references therein).  

We order parameters of \eqref{gsG} in some manner and denote by $n$  the number of parameters of the  system.
Then there is a one-to-one correspondence between the set of all  systems \eqref{gsG}
and the points in the space $\C^n$.
It is well known  that the set of 
all systems \eqref{gsG}  admitting an analytic first integral in a neighborhood of the origin form an affine algebraic variety in $\C^n$, which is called the 
integrability variety (or the center variety) of system \eqref{gsG} (see e.g. \cite{RS,RXZ}).

The irreducible decomposition of the integrability variety is important, particularly for to the following two reasons.

First, each component of the variety 
is related to a certain mechanism of integrability. In this paper, we will discuss only two important  mechanisms: the Darboux integrability and time-reversibility.
Actually,  finding  all  methods which ensure the  construction of a local analytic 
first integral of \eqref{gsG} (provided the system has it)  could be considered as a solution to the integrability 
problem for system \eqref{gsG}.

Second, the components of the variety are  related to
studies on Hilbert's 16th problem.
The second part of this problem, which remains unresolved even in the simplest case of the quadratic system, concerns the maximum number and relative position of limit cycles of polynomial vector fields.
The polynomial vector fields with a center or focus type singularity at 
the origin, 
\be \label{uvSys}
\dot u =-v +U(u,v), \quad \dot v =u+V(u,v),
\ee
can be enclosed into  vector fields \eqref{gsG} with $p=q=1$ via a complexification procedure.  It often happens that  for generic points of  each component of the integrability  variety 
the number of limit cycles bifurcating   from the center or focus 
cannot exceed the codimension of the component \cite{Bau,Chr2005}. 
Thus, the study of the structure of the integrability variety is important 
for the research on the local Hilbert's 16th problem.

As it is known, 
duality in mathematics refers to a principle or relationship where two related structures, concepts, or problems can be transformed into each other in a way that reveals deeper insights about their nature. This transformation often involves reversing or interchanging certain aspects, leading to a "dual" version of the original concept. Duality manifests in various areas of mathematics, and its specific meaning and implications can vary significantly across different fields. Some notable examples are dual vector spaces in linear algebra, dual 
spaces in functional analysis, Lagrange duality in optimization etc.

In the present paper, we investigate an intricate duality between linear operators acting on polynomial vector fields depending on parameters and linear operators acting on power series in the space of parameters of these vector fields.  It is clear that any analytic or formal first integral of system \eqref{gsG} is an element in the kernel of the linear differential operator 
$$
D:= \frac{\partial}{\partial x} X(x,y)+\frac{\partial}{\partial y} Y(x,y),
$$ 
so we study the duality  between the  kernel of $D$  and the kernel of 
a  certain differential operator  in the space of formal power series in  parameters of system \eqref{gsG}. We show that there is a one-to-one correspondence between
first integrals of  the form 
\begin{equation} \label{Int}
	\Psi(x, y) =v_{0,0} x^q  y^p + \sum_{i + j > p+q} v_{i - q,j - p} x^i y^j=
	x^q  y^p \left ( v_{0,0}+  \sum_{i + j > p+q} v_{i - q,j - p} x^{i-q} y^{j-p}  \right)
\end{equation}
of system \eqref{gsG} and formal power series in the parameters of \eqref{gsG} annihilated by a certain  differential operator. 

Furthermore, we discuss some properties of  the integrability variety, 
particularly in relation to the invariants of the Lie group 
generated by the matrix of the linear approximation of system \eqref{gsG}. 

Finally, we introduce differential operators on affine algebraic 
varieties and discuss  relations of power series solutions to the differential operators on varieties and local analytic  first integrals of system \eqref{gsG}.

In recent work \cite{PR}, it is shown that  the space of parameters 
of a polynomial vector field is a kind of dual space 
for  computing  Poincar\'e-Dulac normal forms and the computations of  normal forms  can be performed 
working with series defined on the space of parameters. 
The studies in \cite{PR} and in  the present paper  present a clear evidence 
that the space of parameters of a polynomial vector field  is a kind of dual space when we study certain linear operators on vector fields.

\section{Preliminaries} 

Let
 $\N_0$ be  the set of
non-negative integers and  $\N_{-1}$ be the set 
$\{-1\}\cup \N_0$. 
Consider a system of the form \eqref{gsG} written as 
\begin{equation} \label{gs1}
\begin{aligned}
\dot x &= \phantom{-}  p  x - \sum_{(i,j) \in S}
                                  a_{ij}x^{i+1}{y}^{j} = x( p   - \sum_{(i,j) \in S}
                                  a_{ij}x^{i}{y}^{j}) \\
\dot y &=             -q y+\sum_{(i,j) \in S}
                                  b_{ji}x^{j}{y}^{i+1}=y (   -q +\sum_{(i,j) \in S}
                                  b_{ji}x^{j}{y}^{i})
\end{aligned}
\end{equation}
where the dot means the differentiation with respect to $\tau$,  positive integers $p$, $q$ are  mutually prime, and  $ S$ is the ordered finite set 
\be
 S = \{(u_k,v_k) \mid u_k+v_k >0, k = 1, \ldots,
\ell \} 
\subset 
\N_{-1} \times \N_0. 
\ee 
 Notation \eqref{gs1} 
emphasizes that 
we take into account only the nonzero coefficients of the 
polynomials in \eqref{gsG}.   
We denote by 
$$  
      (a,b) = (a_{u_1,v_1},a_{u_2,v_2},\dots,a_{u_\ell,v_\ell},b_{v_\ell,u_\ell},\dots , b_{v_1,u_1})
$$
the
ordered
$2\ell$-tuple  of coefficients of system \eqref{gs1}, 
by
$\Q[a,b]$ the polynomial 
ring in the variables $a_{ij}$, $b_{ji}$
and   by
$\Q[a^\pm,b^\pm]$ the  
ring of Laurent polynomials  in the variables $a_{ij}$, $b_{ji}, a_{ij}^{-1}, b_{ji}^{-1}.  $
The 
parameter space of \eqref{gs1} is  $\C^{2\ell}$.


It is known  that if system \eqref{gs1} has a  first integral 
represented by a power series 
$$ \Psi(x, y) =   \sum_{i + j \ge  0} u_{i ,j } x^i y^j,
$$
which is convergent in a neighborhood of the origin, 
then it has a  first integral of the form  \eqref{Int}.
It is also known  that if system \eqref{gs1} has a formal first integral of the form \eqref{Int},
then it has also analytic first integral of the same form (see e.g. \cite{RS_BMS}). Moreover, the coefficients $v_{ij}$ are polynomially  dependent on the parameters $a$, $b$ of system \eqref{gs1}.

The polynomials $v_{ij}$ in \eqref{Int} have some special properties. 
To describe them we recall the notion of the multigraded ring.


Let 
$\mathsf{S}= k[x_1^\pm, x_2^\pm, \ldots, x_n\pm]$ be a Laurent  polynomial ring  over a field $k$ and $A$ be  an abelian group.
The  ring $\mathsf{S}$ is called multigraded by $A$   
when
it has been endowed with a degree map $\deg : Z^n \to  A$. In another words, one can say 
that the grading is  given by an exact sequence 
$$
0\to \mathsf{L} \to \Z^n \to A
$$
of abelian groups, where the lattice $\mathsf{L}$ is the kernel of the degree map.

For ${\bf a} \in  A$  let $R_{\bf a}$ denote the vector space (over the field $k$) of 
polynomials of degree $\bf{a}$ in the $A$-grading.
Then $\mathsf{ S}$  has the direct sum decomposition
$$\mathsf{S} = \bigoplus_{{\bf a}\in A} R_{\bf{a}} $$
and it holds 
 $R_{\bf{a}}  R_{\bf{b}} \subseteq R_{\bf{a}+\bf{b}}.$
Elements of  $R_{\bf{a}}$ are called homogeneous polynomials of degree $\bf a$.

We will work with polynomials in variables which are parameters of system \eqref{gs1}. 
Any monomial in  indeterminates  $a_{ij}, 
b_{ji}$  has the form
$$ [ \mu ;\nu] :=\prod_{(i,j)\in S} a_{ij}^{\mu_{ij} }b_{ji}^{\nu_{ji}},
$$
where for all $(i,j)\in S$ we have  $\mu_{ij}$, $\nu_{ji}\in \N_0$ in the case of usual polynomials,
and $\mu_{ij}$, $\nu_{ji}\in \Z$ in the case of Laurent  polynomials.

Once the $\ell$--element set $S$ has been specified and ordered, we let 
$L : \Z^{2\ell} \to \Z^2$ be the map defined by
\be 
\begin{aligned} \label{L1}
L(\mu,\nu) &= (L_1(\mu,\nu), L_2(\mu,\nu))^\top \\
    &=\sum_{(i,j)\in S} \left( { i \ch  j}  \mu_{ij}  + 
    { j\ch i} \nu_{ji} \right),
\end{aligned}
\ee
where $\mu$ and $\nu$ are $\ell$-tuples with the entries from $\Z$. 

Along with the set $S$ we will use the set 
  $\mathcal{S}= S \cup\ \{(j,i) :\,  (i,j)\in S\}$. 
We also agree on convention $a_{ij}=b_{ji}=0$ for $ (i,j) \not  \in S $. 

Denote 
by $\mc Q$ the subgroup of $\Z^2$ generated by $\mathcal{S}$. We will treat  $\Q[a^\pm,b^\pm]$ and  $\C[a,b]$
as the rings  multigraded by 
$\mathcal{Q}$ with the degree map being the map $L(\mu,\nu)$ defined by \eqref{L1}.
%
%
%
%
For this multigrading  it holds that 
$$
 \C[a^\pm,b^\pm]=\bigoplus_{(i,j)\in \mc Q }R^\pm_{(i,j)} 
$$
and 
$$
 R^\pm_{(j,k)}R^\pm_{(s,t)}\subseteq R^\pm_{(j+s,k+t)}.
$$
We  call  elements of $R^\pm_{(u,v)} $ the  $(u,v)$-polynomials
(they are homogeneous of degree $(u,v)$).


Let $M^\pm$ be the set of all formal  power series of the form \eqref{Int} such that 
the coefficient of $x^i y^j$ in \eqref{Int} is an  $(i-q,j-p)$-{polynomial} and let 
$${\mc Q}'=\left( \N_{0}^2 - (q,p) \right) \cap {\mc Q}.$$
Denoting $M^\pm_{(i,j)}=R^\pm_{(i,j)}x^{i+q}y^{j+p}$ we have
$$
M^\pm=
\bigoplus_{
(i,j)\in {\mc Q}'
}
M^\pm_{(i,j)}  
.$$
We define the multiplication of elements $\Psi\in M^\pm$ by elements of $ \C[a^\pm,b^\pm]$ as follows.
Every polynomial $f  \in \C[a^\pm,b^\pm]  $ can be written in the form 
$$
f=\sum_k f_{i_k,j_k},
$$ 
where $f_{i_k,j_k}$ are $(i_k,j_k)$ polynomials. Then we define
$$
f * \Psi =  \Psi  \sum_k f_{i_k,j_k} x^{i_k} y^{j_k},\ \ \Psi\in M^\pm. 
$$
With this operation of multiplication  $
M^\pm
$
is a multigraded $   \C[a^\pm,b^\pm] $-module. Observe that the graded ring $    \C[a^\pm,b^\pm] $ is defined using only the nonlinear terms of \eqref{gs1}, but the module $M^\pm$ is defined also using the linear part of  \eqref{gs1} since the structure of the series \eqref{Int} depends on $p$ and $q$.

In the definitions and constructions above, if we restrict our consideration  to polynomials in the ring $ \Q[a,b]$, we denote the corresponding objects as \(R_{(i,j)}\),  \(M_{(i,j)}\), and \(M\). In this context, it is clear that \(R_{(0,0)}=\Q\).


We mention some properties of the ring $\Q[a,b]$ multigraded by $\mathcal{Q}$.
\begin{pro}\label{pro:mult_grad}
The multigraded ring $\Q[a,b]$ has the following properties:\\
1. The only polynomials of degree $0$ are the constants.\\
2. For all ${(r,s)} \in  \mathcal{Q})$, the $\Q$-vector space $R_{(r,s)}$ is finite-dimensional.\\
3. The only nonnegative vector in the lattice $\ker L$ is $0$.\\
\end{pro}

\begin{proof}
We first prove the third statement. Assume $L(\mu,\nu)=(0,0).$ Then 
\be \label{l1pl2}
 L_1(\mu,\nu)+L_2(\mu,\nu)=\sum_{(i,j)\in S} (  (i+j) \mu_{ij} + (j+i) \nu_{ji}  )=0.
 \ee
Since by our assumption for any $(i,j)\in S$ 
it holds that $i+j> 0$,   equality \eqref{l1pl2} can hold only 
if all entries  of $\mu $ and $\nu $ are zeros or at least one entry  is negative. 

The correctness of the other statements of the proposition follows from
Theorem 8.6 of \cite{MS}.
\end{proof}


\section{Normal forms of power series and necessary conditions of integrability}

In this section, we describe a procedure to compute a power series which 
gives the necessary conditions for the existence of an analytic first integral in a neighborhood of the origin of system \eqref{gs1}. 
We call such a series the series in the standard form. 

We write system \eqref{gs1} in the form 
\begin{equation} \label{gs}
	\begin{aligned}
		\dot x &= \phantom{-} x ( p   - \sum_{(i,j) \in S}
		a_{ij}x^{i}{y}^{j}) =x(p +\widetilde P(x,y)  )=  P(x,y) , \\
		\dot y &=           y (-q+\sum_{(i,j) \in S}
		b_{ji}x^{j}{y}^{i})=y (-q+\widetilde Q(x,y)) =  Q(x,y).
	\end{aligned}
\end{equation}

The  operator of the differentiation with respect to the vector field \eqref{gs} can be written as 
$$
{\mc D}:= {P}\frac{\partial}{\partial x} +{Q}\frac{\partial}{\partial y}={\mc D}_1+{\mc D}_2, 
$$
where 
\be 
{\mc D}_1:=p x \frac{\partial}{\partial x}   -qy \frac{\partial}{\partial y} \qquad 
{\mc D}_2:= x \widetilde P \frac{\partial}{\partial x}   +y\widetilde Q \frac{\partial}{\partial y}.
\ee

We set $v_{k_1,k_2}=0$ if $(k_1,k_2)\not \in {\mc Q}'$
and 
introduce a formal power series $\Psi$  as 
\be \label{e:Psi}
\Psi = x^q y^p  \sum_{(s_1,s_2)\in \mc{Q}'} v_{s_1s_2}x^{s_1}y^{s_2},
\ee 
where $ v_{s_1s_2}$ depend on the parameters  of  system \eqref{gs1}.

A simple computation shows that the action of ${\mc D}_1$ on the monomial  $x^r y^s$ is as follows
\be \label{diag}
{\mc D}_1(x^ry^s)=(p r - q s)x^ry^s.
\ee
Thus, 
\be \label{D1}
{\mc D}_1(\Psi)= x^q y^p  \sum_{(k_1,k_2)\in \mc{Q}'} (k_1 p-k_2 q) v_{k_1k_2}x^{k_1}y^{k_2}.
\ee
Then ${\mc D}_1(x^r y^s)=0$ if and only if $(r,s)=(kq,kp)$ for some $k\in\N_0$. In particular, ${\mc D}_1 (x^q  y^p)=0$.

We recall notation \eqref{e:Psi} and write 
\begin{align} \label{Psi}
	\Psi(x, y) &=x^qy^p\psi,\\
	\psi&=\sum_{(s_1,s_2)\in\mc Q'} v_{s_1,s_2}x^{s_1}y^{s_2}  .
\end{align}
Then, as ${\mc D}_1 (x^q  y^p)=0$, we have
\begin{align}
	D(\Psi)&={\mc D}_1 (x^q  y^p)\psi + x^q  y^p{\mc D}_1(\psi) + {\mc D}_2 (\Psi)\\
	&=x^q  y^p{\mc D}_1(\psi) + {\mc D}_2 (\Psi).\label{e:d}
\end{align}
Furthermore, as 
$$x\frac{\partial  x^{q+s_1} y^{p+s_2}}{\partial x} =(q+s_1)x^{q+s_1}  y^{p+s_2}$$
and 
$$y\frac{\partial  x^{q+s_1} y^{p+s_2}}{\partial y} =(p+s_2)x^{q+s_1}  y^{p+s_2}, $$
we compute
\be  \label{D2psi}
{\mc D}_2 (\Psi)=x^q  y^p \sum_{\substack{(i,j)\in {\mc S} \\(s_1,s_2)\in {\mc Q}'}}   (-(s_1+q)  a_{ij}+ (s_2+p) b_{ij})    v_{s_1,s_2} x^{s_1+i} y^{s_2+j}.
\ee 
If $v_{s_1s_2}\in R_{(s_1,s_2)}$, using that $(s_1+q)a_{ij}$, $(s_2+p)b_{ij}\in R_{(i,j)}$, we see that $(-(s_1+q)  a_{ij}+ (s_2+p) b_{ij}))v_{s_1,s_2}\in R_{(s_1+i,s_2+j)}$.    

Since ${\mc D}_1$ is semisimple, we can write 
$$
M=\ker {\mc D}_1 \bigoplus {\  \rm  im\ }  {\mc D}_1.
$$

From the above expressions it is not difficult to see that the following statement holds.
\begin{pro}
For any $\Psi \in M_{(k_1,k_2)}$, $(k_1,k_2)\in{\mc Q}'$, we have  ${\mc D}_1 (\Psi)\in M_{(k_1,k_2)}$. Consequently, $\Psi \in M$ implies ${\mc D}_1(\psi)\in M$. $M$ is invariant also for ${\mc D}_2$ so, both   ${\mc D}_1:M \to M$ and ${\mc D}_2:M \to M$ can be regarded as $\Q$-linear operators on $M$ and 
$$\ker {\mc D}_1= \bigoplus_{k\in \N_0} M_{qk,pk}.$$
\end{pro}

\begin{defin} 1)
We say that the series  $\Psi \in M$  of the form \eqref{e:Psi}
is in the normal form with respect to system \eqref{gs} (or simply in the normal form) if 
$$
{\mc D}(\Psi)\in \ker {\mc D}_1.
$$
2) If  $\Psi \in M$  
is in the normal form and $v_{00}=1,$  $v_{qk,pk}=0$ for all $k\in \N$, we say that  $\Psi \in M$  
is in the canonical normal form. 
\end{defin}

\begin{teo}\label{th1}
Let family \eqref{gs1}  be given. 
There exists a unique series  $\Psi\ne const$ in the   canonical normal form.
More precisely, 
\begin{enumerate}
\item[(a)] ${\mc D}(\Psi)\in \ker {\mc D}_1$, 
 \be \label{eq_g}
 {\mc D}(\Psi)= \sum_{k=1}^\infty g_{qk,pk} (x^q y^p)^{k+1}
 \ee 
and  for every $k \ge 1$,  $\gk\in  R_{(qk,pk)}$;
\item[(b)] for every  $(k,m) \in {\mc Q}'$,
           $v_{km} \in R_{(k,m)}$;
\item[(c)] $v_{00}=1$ and for every $k\ge 1$, $v_{qk,pk}=0$.
\end{enumerate}
\end{teo}

\begin{proof} 

As it is our aim to find a $\Psi\in M$ and $g(x,y)= \sum_{k=1}^\infty g_{qk,pk} (x^q y^p)^{k+1}\in\ker {\mc D}_1$ such that $ {\mc D}(\Psi) = g(x,y)$, in view of \eqref{e:d}, we consider the following equation
\be
x^q  y^p{\mc D}_1(\psi)= -{\mc D}_2 (\Psi) +g(x,y).
\ee
Recalling \eqref{D1}, \eqref{D2psi} and canceling out  $x^q y^p$ we obtain
\begin{equation*}
\begin{aligned}
  \sum_{(k_1,k_2)\in \mc{Q}'} (k_1 p-k_2 q) v_{k_1,k_2}x^{k_1}y^{k_2}&= \sum_{\substack{(i,j)\in S\\(s_1,s_2)\in\mc Q'}}   ((s_1+q)  a_{ij}- (s_2+p) b_{ij})    v_{s_1,s_2} x^{s_1+i} y^{s_2+j}\\ &+\sum_{k=1}^\infty g_{qk,pk} (x^q y^p)^k.
\end{aligned}
\end{equation*} 
Comparing the powers at $x^{k_1}y^{k_2}$, when when $k_1 p-k_2 q\neq 0$,  gives the identity
\be \label{vk1k2Def}
\begin{aligned} 
(k_1 p-k_2 q) v_{k_1,k_2}&=\sum_{\substack{(i,j)\in {\mc S},(s_1,s_2)\in\mc Q'\\ (s_1+i,s_2+j)=(k_1,k_2)} }((s_1+q)  a_{ij}- (s_2+p) b_{ij})    v_{s_1,s_2}\\
&=\sum_{(i,j)\in {\mc S}}
	((k_1-i+q)  a_{ij} - (k_2-j+p) b_{ij})    v_{k_1-i,k_2-j}.
\end{aligned}
\ee
 Otherwise, by choosing $v_{qk,pk}=0$, $k\ge 1$, and recalling $v_{k_1,k_2}=0$ if $(k_1,k_2)\notin \mc{Q}'$, we have
\begin{equation*} 
g_{qk,pk}=-\sum_{(i,j)\in {\mc S}}
	((q(k+1)-i)  a_{ij}- (p(k+1)-j) b_{ij})    v_{qk-i,pk-j}.
\end{equation*} 
We observe that \eqref{vk1k2Def} defines the sequence $v_{k_1,k_2}$ recursively (since  $i+j>0$ for $(i,j)\in S$). Setting $v_{00}=1$ gives $v_{00}\in R_{(0,0)}$.  By induction, we easily get that $v_{k_1,k_2}\in R_{(i,j)}R_{(k_1-i,k_2-j)}\subset R_{(k_1,k_2)}$ and similarly, $g_{qk,pk}\in R_{(i,j)}R_{(qk-i,pk-j)}\subset R_{(qk,pk)}$ as desired. Therefore,  $g(x,y) \in M$, $\Psi \in M$ and ${\mc D}\Psi=g\in\ker {\mc D}_1$.

\end{proof}

\begin{remark}
{\rm 
 In some aspects, the above theorem is reproving results of  \cite{RS_BMS} with more transparent  reasoning.} 
\end{remark}

To finish this section, we explain our choice of terminology. 
The series  $\Psi$ defined in Theorem \ref{th1} is sometimes called the Lyapunov function. However, it is indeed a Lyapunov function only in the case of a $1:-1$ resonant system \eqref{gs}, which is a complexification of system \eqref{uvSys} (see e.g. \cite[Theorem 6.2.3]{RS}).
Recall that 
to compute the  Poincar\'e-Dulac normal form  of a system of ODEs 
with a non-zero  diagonal matrix of the linear part 
we split the space  of formal vector fields into the kernel and 
the image of the homological operators, which is defined by the linear 
part of the system. Then the vector field  in the normal form 
is a vector field from the kernel of the homological operator.
The  procedure for computing the map in the normal form is somewhat 
similar. We use the operator ${\mc D}_1$ defined by the linear part of system \eqref{gs}  to split $M$ into the direct sum of  $\ker {\mc D}_1$ and ${\rm  im \ }{\mc D}_1$,  and then compute a power series  such that its derivative 
with respect to the vector field \eqref{gs} is in $\ker {\mc D}_1$.
In view of this observation, we believe it is reasonable to call the 
series  $\Psi$ defined in Theorem \ref{th1} the series  in the normal form. 

\section{The dual problem}

{Let 
 ${\mc V} $  be the direct sum of  
 $ R_{(k_1,k_2)}$ for all $(k_1,k_2)\in {\mc Q}'$.
 The elements of  ${\mc V} $ can be considered 
 as formal power series which  we write as 
\begin{equation*} 
V=\sum_{(i,j)\in\mc Q'}V_{ij},
\end{equation*}
where $ V_{ij}\in R_{(i,j)} $ and, for any fixed $(i,j)\in\mc Q'$,
$$ 
V_{ij}=\sum_{L(\mu,\nu)=(i,j)} \kappa(\mu,\nu) \monom, \quad  
 \kappa(\mu,\nu)   \in \mathbb{C}. 
$$ 
We next define a $\Q$-linear map
\be \label{eq:T}
\mathfrak{T}:M \to \mc {\mc V}
\ee 
by 
$$
\mathfrak{T}(\Psi(x,y))= \Psi(1,1).
$$
Clearly, $\mathfrak T$ is a vector space isomorphism   with the inverse map defined by
\begin{align*}
	\mathfrak{T}^{-1} (V)&= \sum_{(i,j)\in\mc Q'}\sum_{L(\mu,\nu)=(i,j)} \kappa(\mu,\nu) \left[\mu;\nu\right]x^{i+q} y^{j+p}\\
	&=
 \sum_{(i,j)\in\mc Q'}\sum_{L(\mu,\nu)=(i,j)}  \kappa(\mu,\nu) \left[\mu;\nu\right]x^{L_1(\mu,\nu)+q} y^{L_2(\mu,\nu)+p}.
\end{align*}
Since $\mathfrak{T}(M_{(i,j)})\subset  R_{(i,j)}$,
we  see that $\mathfrak T$ 
is isomorphism of the graded modules $M$ and $\mathcal{V}$, so it is the so-called   graded isomorphism of degree $(0,0)$.

Let 
$ P^*(a,b)= \mathfrak{T}(P(x,y)) $, $ Q^*(a,b)= \mathfrak{T}(Q(x,y)) $, 
	\be\label{Aop}
{\mc A}(V)=\sum_{(i,j)\in {  {\mc  S}} } ((P^*(a,b)  i+ Q^*(a,b) j){\mc D}_{ij} (V)
+ (q P^*(a,b)+p Q^*(a,b) )V, 
	\ee 
	where for all $(i,j)\in {\mc S}$,
	\begin{equation*}
	    \begin{aligned}
	{\mc D}_{ij}&:{\mc V}\to {\mc V},\\
	 {\mc D}_{ij}&=a_{ij}\frac{\partial}{\partial a_{ij} }+b_{ij}\frac{\partial}{\partial b_{ij}, },\ 
	\end{aligned} 
	\end{equation*}
	are  differentiation operators.

\begin{teo}\label{th2}
The following diagram is commutative.
\begin{center}
\begin{tikzcd}
M \arrow[r, "{\mc D}"] \arrow[d, "\mathfrak{T}"]
&  {M} \arrow[d, "\mathfrak{T}" ] \\
{\mc V} \arrow[r,  "{\mc A}"]  
&   {\mc V }
\end{tikzcd}
\end{center}
\end{teo}
\begin{proof}
	To each of the operators $\mc {\mc D}_1$, $\mc {\mc D}_2: M\to M$ we will assign a formal differential operator ($\Q$-endomorphism) $\mc A_i$, $i=1,2$,  on $\mc V $ in  view of the isomorphism $\mathfrak{T}$, (see \eqref{eq:T}), such that
	\be  \label{eq:DA}
	\mathfrak{T}\circ {\mc D}_j = {\mc A}_j \circ \mathfrak{T},\ \ j=1,2,
	\ee 
	and, consequently 
	\be
	 {\mc A}:= {\mc A}_1+ {\mc A}_2=\mathfrak{T}\circ ({\mc D}_1+{\mc D}_2)\circ \mathfrak{T}^{-1}.
	\ee
Let  operators $\mc A_1$, $\mc A_2$ act  as follows. For any $V\in \mc V$,
	\be \label{A1}
	\mc A_1 (V) = \sum_{(i,j)\in {\mc S}} ( p i- q j){\mc D}_{ij}(V),
	\ee
	\be 
	\mc A_2 (V) = \left(\sum_{(i,j)\in {\mc S}} (\as i+ \bs j){\mc D}_{ij} + q P^*(a,b) + p Q^*(a,b) \right)V.
	\ee
	
We proceed by proving \eqref{eq:DA}. It is enough to check the equality  for every polynomial of the form $v_{ij}x^{q+i}y^{p+j}$, where  $v_{ij}\in R_{(i,j)}$. Then,
	 by \eqref{diag},
\be \label{eq:TD1}
	 \mathfrak{T}\circ {\mc D}_1(v_{ij}x^{q+i}y^{p+j})=(pi-qj)v_{ij}.
\ee
Let $(i,j)\in \mc Q'$ be fixed and write $v_{ij}=\sum_{L(\mu,\nu)=(i,j)} \kappa(\mu,\nu)\monom$. Then, for any $(r,s)\in S$,
\be
\begin{aligned} \label{eq:Drsvij}
{\mc D}_{rs}v_{ij}&=\sum_{L(\mu,\nu)=(i,j)} \kappa(\mu,\nu){\mc D}_{rs}\monom \\
				  &= \sum_{L(\mu,\nu)=(i,j)} (\mu_{rs}+\nu_{rs})\kappa(\mu,\nu)\monom.
\end{aligned}
\ee 
Hence, 
\begin{align*}
	{\mc A}_1\circ \mathfrak{T}(v_{ij}x^{q+i}y^{p+j})&={\mc A}_1v_{ij}\\
	&= \sum_{(r,s)\in {\mc S}} ( p r- q s){\mc D}_{rs}v_{ij}\\
	&= \sum_{(r,s)\in {\mc S}} \sum_{L(\mu,\nu)=(i,j)} ( p r- q s) (\mu_{rs}+\nu_{rs})\kappa(\mu,\nu)\monom.
\end{align*}
Recall that $\sum_{(r,s)\in {\mc S}}r(\mu_{rs}+\nu_{rs})=L_1(\mu,\nu)=i$ and $\sum_{(r,s)\in {\mc S}}s(\mu_{rs}+\nu_{rs})=L_2(\mu,\nu)=j$ so,
	$${\mc A}_1\circ \mathfrak{T}(v_{ij}x^{q+i}y^{p+j})=(pi-qj)v_{ij},$$
which  equals to \eqref{eq:TD1}. Thus, ${\mc A}_1$
satisfies \eqref{eq:DA}.

We repeat a similar reasoning  for ${\mc A}_2$. Applying \eqref{D2psi} we compute
\begin{align*}
	 \mathfrak{T}\circ {\mc D}_2(v_{ij}x^{q+i}y^{p+j})
	 &=\mathfrak{T}\sum_{(r,s)\in {\mc S}}(-a_{rs}(q+i)+b_{rs}(p+j))v_{ij}x^{q+r+i}y^{p+s+j}\\
	 &=\sum_{(r,s)\in {\mc S}}(-a_{rs}(q+i)+b_{rs}(p+j))v_{ij}\\
	 &=(\as (q+i)+\bs (p+j))v_{ij},
\end{align*}
while by application of \eqref{eq:Drsvij} and \eqref{L1}
\begin{align*}
	{\mc A}_2\circ \mathfrak{T}(v_{ij}x^{q+i}y^{p+j})&={\mc A}_2(v_{ij})\\
	&=  \sum_{(r,s)\in {\mc S}} (\as r+ \bs  s){\mc D}_{rs}(v_{ij}) +\\
&\phantom{=} ( q \as +p \bs )v_{ij} \\
	&= \sum_{(r,s)\in {\mc S}} ( \as r+ \bs s) 
\sum_{L(\mu,\nu)=(i,j)} (\mu_{rs}+\nu_{rs})\kappa(\mu,\nu)\monom
\\ & \phantom{=} 	+( q \as + p \bs ) v_{ij}\\
	&=\sum_{L(\mu,\nu)=(i,j)} ( \as L_1(\mu,\nu) +  \bs  L_2(\mu,\nu)) \kappa(\mu,\nu)\monom
\\ & \phantom{=}	+(q \as + p \bs ) v_{ij}\\
	&=((q+i)\as + (p+j) \bs ) v_{ij}
\end{align*}
as desired. 
\end{proof}

\begin{remark}
    {\rm In the case of $1:-1$ resonant system \eqref{gs} the dual operator $\mc A$ was introduced in \cite{RR} using a different approach. }
\end{remark}

We order the  elements of  ${\mc Q}'$  using the degree lexicographic term order, so here and below for $V \in \mc V$
the notation 
$V=V_{km}+h.o.t. $
means that $V$ is of the form
$$
V=V_{km}+\sum_{(i,j)>(k,m)} V_{ij},
$$
where $ V_{km}\in R_{(k,m)}, \  V_{ij}\in R_{(i,j)}.  $

\begin{teo}\label{propAV}
If system  \eqref{gs} has an analytic or formal  first integral, then the  equation
\be
\label{eq:AV}
\mathcal{A}(V)=0
\ee
 has a formal 
power series solution $V=1+h.o.t.\in \mc V$. 
\end{teo}
\begin{proof}
By the results of \cite{BW,Zung}  if system \eqref{gs} has a local 
 analytic or formal first integral, then it can be transformed to 
a Poincar\'e-Dulac normal form, which admits a non-constant formal  first integral.
Then, by Lemma 3 of \cite{BW} the normal form satisfies Bruno's Condition A,
so it has the first integral  $x^q y^p$.    
Going back to the original coordinates we see that system \eqref{gs} has a first integral $\Psi(x,y)$  of 
the form \eqref{Int}. Thus, by Theorem \ref{th2} $\mathfrak{T}(\Psi)$ is a solution to \eqref{eq:AV}.
\end{proof}

{\it 
From the above  theorems we see that the space $\mathcal{V}$ is a kind of dual space to $M$ and 
the operator $\mc A$ is dual to $\mc D$.}

Note that
to find a solution to  equation \eqref{eq:AV} on $\C^{2\ell}$
 one can try to use the  Lagrange-Charpit system associated with  equation \eqref{eq:AV},
\be \label{Vvf}
\begin{aligned}
\dot a_{ij}  =& a_{ij} (\ps  i+ \qs  j), \ (i,j)  \in {\mc S},\\
\dot b_{ij}  =& b_{ij} ( \ps  i+  \qs j), \ (i,j)  \in {\mc S},\\
\dot w\phantom{_{ij}}  =& -w
 (\pqs  ).
\end{aligned}
\ee
We call system \eqref{Vvf} the  (non-homogeneous) dual system and  
\be \label{Vvf_sh}
\begin{aligned}
\dot a_{ij}  =& a_{ij} (\ps  i+ \qs j), \ (i,j)  \in {\mc S},\\
\dot b_{ij}  =& b_{ij} (\ps  i+ \qs j), \ (i,j)  \in {\mc S},
\end{aligned}
\ee
the homogeneous dual system (since it corresponds 
to the homogeneous part of operator $\mc A$).

Denote by ${\mc X}$ the   differential operator 
$$ {\mc X}=\sum_{(i,j)\in S}  \frac{\partial}{\partial a_{ij}} \dot a_{ij} +  \sum_{(i,j)\in S} \frac{\partial}{\partial b_{ji}}\dot b_{ji} - \frac{\partial}{\partial w} \dot w. $$

\begin{pro}
A series  $\Phi$ of the form $\Phi=w V$, where $V\in \mathcal{V}$, is a first integral of system \eqref{Vvf} if and only if $V$ is a solution 
to \eqref{eq:AV}.
\end{pro}
\begin{proof}
Taking into account that ${\mc X}(w)=  
 w (\pqs )$
 we have ${\mc X}(\Phi)= w {\mc X}(V)+ V {\mc X}(w)=w ({\mc X}(V)+(\pqs )V)$.
Since $  {\mc X}(V)+(\pqs ) V={\mc A}(V)$, the claim holds. 
\end{proof}

Denote by $\widehat{\mc A}$ the differential operator 
which appears in \eqref{Aop},  that is, 
\be\label{Aophat}
\widehat {\mc A}(V)=\sum_{(i,j)\in {  {\mc  S}} } ((P^*(a,b)  i+ Q^*(a,b) j){\mc D}_{ij} (V).	\ee 

To finish this section, we give another expression for 
operator \eqref{Aophat}. For 
$$
V=\sum_{\substack{(\mu,\nu): L(\mu,\nu)\in {\mc Q}'}}  
\kappa(\mu,\nu) \left[\mu;\nu\right], \quad  
 \kappa(\mu,\nu)  \in \C,
$$
let 
$$
\mathfrak{D}:{\mc V} \to {\mc V}
$$
be defined by 
$$
\mathfrak{D}({V})=  \sum_{(\mu,\nu):L(\mu,\nu)\in {\mc Q}'}  \kappa(\mu,\nu) 
L(\mu, \nu )  \left[\mu;\nu\right]. 
$$
\begin{remark}
The operator $\mathfrak{D}$ is a particular case of the one introduced in \cite[p.17]{PR}.   
\end{remark}

For $u={u_1 \ch u_2}$, $v=(v_1,v_2)$, we define
$$
\la u,v \ra:= v\cdot u=u_1v_1+u_2v_2.  
$$
\begin{pro} For any $V\in {\mc V}$
 $$ 
 \widehat {A}(V)=\la \mathfrak{D}(V), (\ps,\qs)\ra. $$
\end{pro}
\begin{proof}
 From \eqref{Aophat} we have 
 $$
\begin{aligned}
  \widehat {\mc A}(V)&=\sum_{(i,j)\in S} 
  \la  {i \choose j} , (\ps,\qs) \ra {\mc D}_{i,j}(V)\\
  &=\la   \sum_{(i,j)\in S}  
  {i \choose j}  {\mc D}_{i,j}(V) , (\ps,\qs)  \ra.
\end{aligned}
 $$
 Observe that 
 $$
 {i \choose j}
 {\mc D}_{ij}( \monom )= \left[ {i\choose j} \mu_{ij} +  {i \ch j } \nu_{ij}\right]\monom  .
 $$
 Thus, $$
 \begin{aligned}
   \widehat {\mc A}(V)= &\la    \sum_{(\mu,\nu):L(\mu,\nu)\in {\mc Q}'}  \sum_{(i,j)\in S} \kappa(\mu,\nu) \monom  \left[ {i \choose j} \mu_{ij} +   {i \choose j} \nu_{ij} \right], (\ps, \qs)   \ra \\ 
   = & \la    \sum_{(\mu,\nu):L(\mu,\nu)\in {\mc Q}'}   \kappa(\mu,\nu)  L(\mu,\nu) \monom  , (\ps, \qs)   \ra \\ =  &\la  \mathfrak{D}(V), (\ps, \qs) \ra. 
   \end{aligned} 
$$
\end{proof}

\section{The Bautin ideal and invariants of a Lie group}
\label{BauId}

Let the map $ \Psi(x,y)$ be in the canonical normal form. 
Then $ V= \Psi(1,1)$ has the property
\be\label{Vprop}
{\mc A}( V)=\sum_{k=1}^\infty g_{qk,pk}(a,b), \quad {\rm where } \quad   g_{kq,kp}(a,b)\in R_{qk,pk}, \quad  v_{kq,kp}=0,
\ee
equivalently,
$$
{\mc D}( \Psi(x,y))=\sum_{k=1}^\infty   g_{qk,pk}(a,b) x^{qk} y^{pk}. $$

 Note that equating the coefficients
 of the same monomials on both sides of \eqref{Vprop} we obtain a simple recurrent 
 formula for computing the function $V(a,b)$,
 which 
 yields  an  efficient algorithm for computing 
 the polynomials $g_{qk,pk}$ given in Appendix of \cite{RS_BMS}.

It is not difficult to see that there is only one series $ V(a,b)\in \mc V$  of the form $ V=1+h.o.t.$ satisfying   \eqref{Vprop}, so the 
polynomials $ g_{qk,pk}$, $k\in \N$ are uniquely defined. The ideal 
$$
B=\la g_{qk,pk} \ : k \ \in \N \ra
$$  
is called the {\it Bautin ideal} of system \eqref{gs}. The variety $\vv(B)$ of $B$ is called the {\it  integrability  variety}  (or
the  center variety) of system \eqref{gs} (\cite{RS,RXZ}).
If the values of parameters of system \eqref{gs} belong to $\vv(B)$, then the corresponding system \eqref{gs} admits 
a local first integral of the
 form \eqref{e:Psi}.

In this section we discuss  some properties of the Bautin ideal and the integrability variety
of system \eqref{gs}.

\subsection{Integrability quantities and the Bautin ideal}


It is clear that 
there are infinitely many series 
$\wt  V=1+h.o.t.  \in \mathcal{V}$
and infinitely many corresponding  polynomials  $ \wt g_{kq,kp}(a,b)\in R_{qk,pk}$
such that 
 \be \label{eq:vgt}
{\mc A}(\wt V)=\sum_{k=1}^\infty \wt g_{qk,pk}(a,b).
\ee
The next proposition shows that for any choice of  $\wt  V=1+h.o.t.  \in \mathcal{V}$
and  $\wt g_{kq,kp}(a,b)\in R_{qk,pk}$ satisfying \eqref{eq:vgt} the 
polynomials  $\wt  g_{kq,kp}(a,b)$ define the same variety.

\begin{pro}\label{pro:4}
  Assume that $\wt  V=1+h.o.t.  \in \mathcal{V}$
  and  
$\wt g_{qk,pk}(a,b)\in R_{qk,pk}
$ satisfy \eqref{eq:vgt}.
Let $$
\wt B=\la\wt  g_{qk,pk} \ : k \ \in \N \ra.
$$  
Then $\vv(B)=\vv(\wt B)$.
\end{pro}
\begin{proof}
Let $\wt \Psi(x,y)= \mathfrak{T}^{-1}(\wt V)$.
By Theorem \ref{th2} 
$$
{\mc D}(\wt \Psi)=\sum_{k=1}^\infty \wt  g_{qk,pk}(a,b) x^{qk} y^{pk}. $$
Then by Theorem 3 of \cite{RS_BMS}  $\vv(B)=\vv(\wt B)$.
\end{proof}

The multigrading by $\mathcal{Q}$  depends only on the nonlinear part of system \eqref{gs},
but the structure of the integrability quantities $g_{kq,kp}$ depends also 
on $p$ and $q$, since by Theorem \ref{th1} the polynomial  $g_{kq,kp}$ is homogeneous of degree 
$(kq,kp)$.

For instance, for   the $p:-q$ resonant 
  quadratic system  
 \be\label{vg_pq}
\begin{aligned}
\dot x &= \phantom{-(} 
         p x - a_{10} x^2 - a_{01} xy - a_{-12} y^2  \\
\dot y &=           
          - q y + b_{2,-1} x^2 + b_{10} xy + b_{01} y^2 
\end{aligned}
\ee
the degree map \eqref{L1}  is 
\begin{equation} \label{Lpq}
L(\mu,\nu) = { L^1(\mu,\nu) \choose L^2(\mu, \nu)}
       = {1 \choose 0} \mu_1 
       + {0 
\choose 1} \mu_2
       + {-1 \choose 2} \mu_3
       + {2 \choose -1} \nu_1  + {1 \choose 0} \nu_2
       +  {0 
\choose 1} 
\nu_3.
\end{equation}
The corresponding group $\mc Q$ is generated by the column vectors in display \eqref{Lpq} and, in this case, coincides with $\Z^2$.

By Proposition \ref{pro:mult_grad} for  all {$k \in L(\N_0^{2\ell})$}   the $\Q$-vector space $ R_{kq,kp}$ is 
finite-dimensional.  We can describe the structure of $ R_{kq,kp}$ and of  the integrability quantities 
$g_{kq,kp}$ more precisely using the invariants of a certain group  related to \eqref{gs}.

\subsection{Invariants of a Lie group }

Let $k$ be a field and $G$ a subgroup of the multiplicative group of invertible  $n \times n$ matrices with elements in $k$. 
For a matrix $A \in G$, 
let $A {\bf x}$ denote the usual action of $G$ on  $k^n$. 
A polynomial $f \in \kxn$ is  \emph{invariant under the action of the group $G$} (or simply \emph{an invariant of $G$}) if $f({\bf x}) = f(A {\bf x})$  for every ${\bf x} \in k^n$ and every $A \in G$.

The matrix
${\rm diag}(p,-q)$
of the  linear part of \eqref{gs}  
is an infinitesimal generator of the one-parameter  group $G$  of transformations 
\be \label{xyT}
x\mapsto e^{p \phi} x, \quad y\mapsto e^{-q \phi} y.
\ee 
The action of this group on system \eqref{gs} yields the change 
of parameters 
\be \label{abT}
a_{ij}\mapsto e^{(-p i+q j)\phi  } a_{ij}, \quad b_{ji}\mapsto e^{(-p j+q i)\phi  } b_{ji},
\ee
so \eqref{abT} is the representation of \eqref{xyT} in the space of parameters of system \eqref{abT}
(in the case $p=q=1$ polynomial invariants of  \eqref{abT} were studied in \cite{Liu2,Sib1,Sib2}).
We denote by $\C[a,b]^G$ the  {ring of invariants of $G$, also called the invariant ring.}

The next statement gives   intrinsic relations 
between  the  invariants of group \eqref{abT}, 
the degree of invariants and the first integrals of the 
linearization of \eqref{Vvf}.

\begin{teo}\label{th4}
1)
The monomial $\monom$ is an invariant of the group defined in \eqref{abT}
 if and 
only if it is of  degree
 $(qk,pk)$ for some $k\in \N_0$.  \\
2)
If a monomial $\monom$ is an invariant of group \eqref{abT},
 then   $\monom$ is a first integral of system
\be \label{Vvf_L}
\begin{aligned}
\dot a_{ij}  =&  (p  i -q j)   a_{ij}  , \ (i,j)  \in S,\\
\dot b_{ji}  =& (p  j -q j) b_{ji} , \ (i,j)  \in S
\end{aligned}
\ee
(which is the linearization of the dual homogeneous system \eqref{Vvf_sh}).
\end{teo}
\begin{proof}
1) Under the action of  \eqref{abT} the monomial 
$\monom$ is changed to 
$$ e^{-pL_1(\mu,\nu)+qL_2(\mu,\nu)}
[\monom.$$
Thus,  $\monom$ is invariant if and only if  
\be\label{eq:conL_1}
 pL_1(\mu,\nu)-qL_2(\mu,\nu)=0.
\ee
Since $i+j > 0$ for all $(i,j)\in S$,
the latter equality is equivalent to 
\be \label{eq:conL_2}
L_1(\mu,\nu)=qk, \ L_2(\mu,\nu)=pk 
\ee 
 for some $k\in \N_0$.

Statement 2) can be verified by a straightforward computation and its proof is omitted.
\end{proof}


Denote by $\mathfrak{M}$ the $1\times 2\ell$  matrix
of equation \eqref{eq:conL_1}, 
\begin{equation} \label{lsib_M}
\mathfrak{M}=(\mathfrak{M}_1 \quad  \mathfrak{M}_2 ),
\end{equation}
where
$$ \mathfrak{M}_1=
\left( (p i_1 -q j_1) \ 
\dots \ (p i_\ell-q j_\ell)\right), \quad \mathfrak{M}_2=\left(  
 (p j_\ell- q i_\ell) \ 
 \dots\ (p j_1-q i_1)\right).
$$
Since \eqref{eq:conL_1} and \eqref{eq:conL_2} are equivalent,  the monomial $\monom$ is of degree  $(kq,kp)$ for some $k\in \N_0$ if and only if  $\mathfrak{M} \cdot (\mu,\nu)^\top=0$, that is $ (\mu,\nu) $ is an element of the monoid 
\be \label{monoidM}
{\mc M}^{(p,q)}=\{ (\mu,\nu) \ : \ \mathfrak{M} \cdot (\mu, \nu)^\top=0  \}.
\ee}


Recall that for a given $(d\times n)$-matrix $\mathfrak{A}=[\mathfrak{a}_1, \dots, \mathfrak{a}_n ]$
the matrix group $$
\Gamma_{\mathfrak{A}}={\rm diag}[  t^{\mathfrak{a}_1}  \ t^{\mathfrak{a}_2} \  
  \ 
 \dots\ t^{\mathfrak{a}_n}],  \quad t\in (\mathbb{C}^*)^d , 
$$
 is isomorphic to the group  $ (\mathbb{C}^*)^d  $ of invertible diagonal
$d \times d$-matrices, which is the $d$-dimensional algebraic torus. Now $\Gamma_{\mathfrak{A}}$
is said to be
the torus defined by $\mathfrak{A}.$ 

With the matrix $\mathfrak{M}$  we associate the  group 
$$
\Gamma_{\mathfrak{M}}={\rm diag}[  t^{p i_1 -q j_1} \ 
\dots \ t^{p i_\ell-q j_\ell} \  
 t^{p j_\ell- q i_\ell} \ 
 \dots\ t^{p j_1-q i_1}], \ \ t\in \C^*,
$$
where $\C[a,b]^{  \Gamma_{\mathfrak{M}} }$ is the invariant ring of $\Gamma_{\mathfrak{M}}$.

\begin{pro}\label{pro:5}
1)  $\C[a,b]^G= \C[a,b]^{\Gamma_{\mathfrak{M}}} $ 
and  a finite set 
$(\mu_1,\nu_1), \dots, (\mu_m,\nu_m) $ is a Hilbert basis of $\mathcal{M}^{(p,q)}$ if and only if 
  $\C[a,b]^{  G }$ 
is minimally generated as a $\C$-algebra by $[\mu_1;\nu_1], \dots, [\mu_m;\nu_m] $.\\
2)
The basis of  $ R_{kq,kp}$  as a $\Q$-vector space  is 
$$
    \beta_{kq,kp}=  \{ \monom \ : \ (\mu,\nu) \in \mathcal{M}^{(p,q)} \ {\rm and} \ L(\mu,\nu)=  (kq,kp) \}.
$$
3) For any $k\in \N_0$ we have that  $g_{kq,kp}\in \Q[a,b]^G$ and are spanned as vectors by elements of $\beta_{kq,kp}$. 
\end{pro}
\begin{proof}
Since the group $\Gamma_{\mathfrak{M}}$ is the same as group \eqref{abT}, the first statement follows 
from Lemma 1.4.2 of \cite{St-AIT}.

The second and third statements are obvious. 
\end{proof}

The Hilbert basis of ${\mc M}^{(p,q)} $ 
can be easily obtained   using Algorithm 1.4.5 of
 \cite{St-AIT} from the toric ideal  $I_ {\Lambda(\mathfrak{M})}$  of 
 the Lawrence lifting of $\mathfrak{M}$ defined by  
$$
\Lambda(\mathfrak{M})= \begin{pmatrix} \mathfrak{M} & 0\\ E_{2\ell} & E_{2\ell} \end{pmatrix},
$$
where  $E_{2 \ell}$ stands for  the  $2 \ell \times 2\ell $ identity matrix. 

\begin{exstar} {\rm Consider the quadratic system \eqref{vg_pq} with $p=1$ and $q=2$,
 \be\label{qv_12}
\begin{aligned}
\dot x &= \phantom{-} 
          x -a_{10} x^2 - a_{01} xy - a_{-12} y^2  \\
\dot y &=   
         -2 y + b_{2,-1} x^2 + b_{10} xy + b_{01} y^2 \,.
\end{aligned}
\ee
In this case the matrix $\mathfrak{M}$ of \eqref{lsib_M} is
$$ 
\mathfrak{M}=[ 1\ -2 \ -5\ 4\ 1 \ -2 ]. 
$$
Following   Algorithm 1.4.5 of
 \cite{St-AIT} 
we compute the Groebner basis of the ideal 
$$ 
 \la   a_{10} - y_1 t,  a_{01}    - y_2 t^{-2}  ,  a_{-12} - y_3  t^{-5}, 
  b_{2,-1} - y_4 t^4,  b_{10} - y_5 t,  b_{01} - y_6 t^{-2} \ra  
 $$
 using an elimination term order with 
 $ \{t\}\succ  \{ a_{10}, \dots, b_{01}  \}  \succ \{ y_1, \dots, y_6\} $
 obtaining the list 
of binomials 
 \begin{multline*}
 \begin{gathered}
      b_{01}  b_{10}^2 - y_5^2 y_6, 
 b_{01}^3  b_{2,-1} - y_4 y_6^3, - a_{-12}  b_{10}^5 + y_3 y_5^5, - a_{-12}  b_{01}  b_{10}  b_{2,-1} + 
 y_3 y_4 y_5 y_6,\\ - a_{-12}^2  b_{10}^4  b_{2,-1} + y_3^2 y_4 y_5^4, - a_{-12}^2  b_{01}  b_{2,-1}^2 + 
 y_3^2 y_4^2 y_6, - a_{-12}^3  b_{10}^3  b_{2,-1}^2 + 
 y_3^3 y_4^2 y_5^3, - a_{-12}^4  b_{10}^2  b_{2,-1}^3 + \\
 y_3^4 y_4^3 y_5^2,  - a_{-12}^5  b_{10}  b_{2,-1}^4 + y_3^5 y_4^4 y_5, - a_{-12}^6  b_{2,-1}^5 + 
 y_3^6 y_4^5, - a_{01}  b_{10}^2 + y_2 y_5^2,  - a_{01}  b_{01}^2  b_{2,-1} + 
 y_2 y_4 y_6^2,\\ - a_{01}  a_{-12}  b_{10}  b_{2,-1} + y_2 y_3 y_4 y_5, - a_{01}  a_{-12}^2  b_{2,-1}^2 + 
 y_2 y_3^2 y_4^2, - a_{01}^2  b_{01}  b_{2,-1} + y_2^2 y_4 y_6, \\  - a_{01}^3  b_{2,-1} + 
 y_2^3 y_4,  - a_{10}  b_{01}  b_{10} + y_1 y_5 y_6,
- a_{10}  a_{-12}  b_{10}^4 + y_1 y_3 y_5^4, - a_{10}  a_{-12}  b_{01}  b_{2,-1} +\\ 
 y_1 y_3 y_4 y_6,- a_{10}  a_{-12}^2  b_{10}^3  b_{2,-1} + 
 y_1 y_3^2 y_4 y_5^3, - a_{10}  a_{-12}^3  b_{10}^2  b_{2,-1}^2 + 
 y_1 y_3^3 y_4^2 y_5^2, - a_{10}  a_{-12}^4  b_{10}  b_{2,-1}^3 + \\
 y_1 y_3^4 y_4^3 y_5, - a_{10}  a_{-12}^5  b_{2,-1}^4 + y_1 y_3^5 y_4^4, - a_{01}  a_{10}  b_{10} + 
 y_1 y_2 y_5, - a_{01}  a_{10}  a_{-12}  b_{2,-1} + y_1 y_2 y_3 y_4,\\ - a_{10}^2  b_{01} + 
 y_1^2 y_6, - a_{10}^2  a_{-12}  b_{10}^3 + y_1^2 y_3 y_5^3, - a_{10}^2  a_{-12}^2  b_{10}^2  b_{2,-1} + \\
 y_1^2 y_3^2 y_4 y_5^2, - a_{10}^2  a_{-12}^3  b_{10}  b_{2,-1}^2 + 
 y_1^2 y_3^3 y_4^2 y_5, - a_{10}^2  a_{-12}^4  b_{2,-1}^3 + y_1^2 y_3^4 y_4^3, - a_{01}  a_{10}^2 + \\
 y_1^2 y_2, - a_{10}^3  a_{-12}  b_{10}^2 + y_1^3 y_3 y_5^2, - a_{10}^3  a_{-12}^2  b_{10}  b_{2,-1} + 
 y_1^3 y_3^2 y_4 y_5, - a_{10}^3  a_{-12}^3  b_{2,-1}^2 + 
 y_1^3 y_3^3 y_4^2,\\ - a_{10}^4  a_{-12}  b_{10} + y_1^4 y_3 y_5, - a_{10}^4  a_{-12}^2  b_{2,-1} + 
 y_1^4 y_3^2 y_4, - a_{10}^5  a_{-12} + y_1^5 y_3.
\end{gathered}
 \end{multline*}
 From the obtained list we keep the monomials depending only on the parameters 
 of  system \eqref{qv_12}.
According to the algorithm 
these  monomials (the first monomials of the binomials) form 
 a basis of the $\Q$-algebra $\Q[a,b]^G$ and the exponents of the monomials 
  form the Hilbert basis of the monoid ${\mc M}^{(p,q)} $. 
By  Theorem \ref{th4} they are also  the first integrals of system \eqref{Vvf_L} (corresponding to system \eqref{qv_12}). 
}
\end{exstar}

As we have mentioned in the Introduction,
the most interesting and important systems in the family \eqref{gs} are $1:-1$ resonant 
systems, so now we limit the consideration to  systems 
\begin{equation} \label{gsM}
	\begin{aligned}
		\dot x &= \phantom{-} x ( 1   - \sum_{(i,j) \in S}
		a_{ij}x^{i}{y}^{j}), \\
		\dot y &=         -  y (1-\sum_{(i,j) \in S}
		b_{ji}x^{j}{y}^{i}).
	\end{aligned}
\end{equation}
  Now, the  corresponding integrability 
quantities $g_{k,k}$ have an additional structure.  For a given 
$(\mu,\nu)$ denote by $(\widehat{\mu,\nu})$ its involution under the map $\mu_{ij}\leftrightarrow\nu_{ji}$ and let $\mathcal{M}$ be  monoid \eqref{monoidM}
defined by  matrix \eqref{lsib_M} with $p=q=1$. 
We can write
$$
\mathcal{M}=\bigoplus_{k\in \N_0} \mathcal{M}_k,
$$
{
where $  \mathcal{M}_k=\{ (\mu,\nu)\in \mathcal{M} : L(\mu,\nu)=(k,k) \}$ due to the  equivalence of \eqref{eq:conL_1} and \eqref{eq:conL_2}.
}

By Proposition \ref{pro:5} 
 polynomials  $g_{k,k}$ are polynomials of the monoid ring $\Q[\mathcal{M}]$. 
Moreover, in this case, 
  {$(\mu,\nu)\in \mathcal{M}$ implies  $(\widehat{\mu,\nu})\in 
\mathcal{M}$} and by  Theorem 3.4.5 of \cite{RS} (see also \cite{Sib1,Liu2,JLR}),  the integrability quantities $g_{kk} $ have the form  
\be
\label{gkk}
 g_{k,k} = \sum_{(\mu,\nu)\in {\mc M}_k} \gamma(\mu,\nu) ( \monom - [\widehat{\mu;\nu}]).
\ee

 The ideal 
$$
I_S=\la \monom-[\widehat{\mu; \nu} ] \ : (\mu,\nu)\in \mathcal{M}  \ra
$$
is called the Sibirsky ideal of system \eqref{gs} \cite{JLR}.
From \eqref{gkk} we see that 
   $\vv(I_S)$ is a subvariety of the variety $\vv(B)$ of the Bautin ideal $B$ 
   of system \eqref{gs}, that is, we have 
the following statement.
\begin{pro}\label{pro:6}
    If $(a,b)\in \vv(I_S) $, then the corresponding system \eqref{gsM}
    has an analytic first integral in a neighborhood of the origin.
\end{pro}

\begin{conj}[\cite{JLR}]
The set  $\vv(I_S)$ is a component of the variety $\vv(B)$, that is, it is  a proper irreducible subvariety of $\vv(B)$.
\end{conj}   

\begin{remark} {\rm 
The meaning of the Sibirsky ideal  from the point of view of the geometry of vector fields \eqref{gsM}
is  that  $\vv(I_S)$ is  the Zariski closure in the space of parameters  of 
  the set of time-reversible systems in family  \eqref{gsM} (\cite{JLR,RS}).  
  }
  \end{remark}

Since the ideal $I_S$ is prime and binomial, it is a toric ideal.
By Theorem 3.8  of \cite{GJR}
$I_S$ is  the toric ideal of the matrix 
\be \label{Mfr}
\widetilde{\mathfrak{M}}= \left( \begin{array}{cc} \mathfrak{M}_1 & \mathfrak{M}_2  \\ 
E_{\ell} & \widehat E_{\ell} 
\end{array} \right),
\ee 
where $ \widehat E_{\ell}$
is the $({\ell}\times{\ell})$-matrix having 1 on the secondary diagonal and the other elements equal to 0 \cite{JLR}, and  $ (\mathfrak{M}_1 \  \mathfrak{M}_2) $ is defined by \eqref{lsib_M}. 
The torus defined by the above matrix has dimension $\ell$ if system \eqref{gsM} is 
in the Poincar\'e-Dulac normal form (so $\mathfrak{M}=0$) and $\ell+1$ otherwise.

 
 Let  $\mathfrak{M}\ne 0$. Then from \eqref{Mfr} using  
 Theorem 1.1.17 of \cite{CLS} we see that the variety $\vv(I_S) $ contains the torus $(\mathbb{C}_{\ell+1})^\ell$ as Zariski open subset via the map 
 \be \label{par_tor}
 (t,y_1,\dots y_\ell) \to   ( y_1 t^{p i_1 -q j_1},  \ 
\dots , \ y_\ell  t^{p i_\ell-q j_\ell}, \  
 y_\ell t^{p j_\ell- q i_\ell}, \ 
 \dots , \ y_1 t^{p j_1-q i_1}).
 \ee

\begin{pro}\label{pro:ap}
All associated primes of the Bautin ideal of system \eqref{gs} are multigraded by the group $\mathcal{Q}$.
\end{pro}
\begin{proof}
Obviously, the group $\mathcal{Q}$ is a torsion-free abelian group.
Thus by Proposition 8.11 of \cite{MS} 
 all associated primes of $B$  are multigraded.
\end{proof}

In the other words, Proposition \ref{pro:ap} 
states that the associated primes of the Bautin ideal 
are generated by homogeneous polynomials in the multigraded ring $\C[a,b]$.

We recall the structure of the variety  $\vv(B)$ of the Bautin ideal  in the 
 case of  the quadratic  system
 \be\label{vgExfamily}
\begin{aligned}
\dot x &= \phantom{-(} 
          x - a_{10} x^2 - a_{01} xy - a_{-12} y^2  \\
\dot y &=           
          -(y - b_{2,-1} x^2 - b_{10} xy - b_{01} y^2) \,.
\end{aligned}
\ee
For this system one can   compute three polynomials $g_{11}, g_{22}, g_{33}$ satisfying  \eqref{Vprop} and find that in this case  the variety of
the  ideal  $B= \la  g_{11}, g_{22}, g_{33} \ra $ 
consists of the following  four  irreducible components:
\begin{itemize}
\item $\vv_1 = \vv (J_1)$, \textrm{where}
      $J_1 = \la 2a_{10}-b_{10}, \ 2 b_{01}-a_{01}\ra;$

\item $\vv_2 = \vv (J_2)$, \textrm{where}
      $J_2 = \la a_{01}, \ b_{10}\ra;$

\item $\vv_3 = \vv (J_3)$, \textrm{where}
      $J_3= \la 2 a_{01}+b_{01}, \ a_{10}+2 b_{10},
                                 \ a_{01} b_{10} - a_{-12} b_{2,-1}\ra;$

\item $\vv_4 = \vv (J_4)$, \textrm{where}
      $J_4 =  \la f_1, f_2, f_3, f_4, f_5 \ra$\,,
      \textrm{where}
      \\
       $f_1 = a_{01}^3 b_{2,-1}-a_{-12} b_{10}^3$,
       $f_2 = a_{10} a_{01} - b_{01} b_{10}$,
      $f_3 = a_{10}^3 a_{-12}- b_{2,-1} b_{01}^3$,
        $f_4 = a_{10} a_{-12} b_{10}^2- a_{01}^2 b_{2,-1} b_{01}$,
         $f_5 = a_{10}^2 a_{-12} b_{10} - a_{01} b_{2,-1} b_{01}^2.$
\end{itemize}

   It is known (see e.g. \cite{ChR, RS})
that $B$ is the Bautin ideal of system \eqref{vgExfamily} and $J_1, \dots, J_4$ are its associated primes.
We see that in the agreement with Proposition \ref{pro:ap} all associated primes of $B$
are multigraded by $\mc Q$. 

The component $\vv_1$ consists of Hamiltonian systems, $\vv_4$
of time-reversible systems, and systems from $\vv_2$ and $\vv_3$ are Darboux integrable (see e.g. \cite{Mal,RS,S1,Zqv}). 
We will consider the restrictions of the operator 
\eqref{Aop} to these components in Section 7.




\section{Integrability and the dual operator}

Since 
$$
{\mc A}_1: \mc V \to \mc V
$$
is semisimple (see \eqref{eq:DA}), we have 
$$
{\mc V}= {\rm im\ } {\mc A}_1 \oplus \ker {\mc A}_1.
$$
Clearly,  $  \ker {\mc A}_1 $ is equal to the direct  sum of $R_{qk,pk},\  k\in \N$.

For a set ${\bf W}\subset \C^{2 \ell}$ we denote by 
${\bf I}({\bf W})$  the ideal of the set ${\bf W}$.
Let  $J$ be  a homogeneous ideal in the multigraded ring $\C[a,b]$.
For a $V=\sum_{(i,j)\in {\mathcal{Q}'}} V_{ij} \in \mc V $
we   define
$$   
V/J := \sum_{(i,j)\in {\mathcal{Q}'}} V_{ij}/J,  
$$
where $V_{ij}/J $ is the quotient of the 
polynomial $V_{ij}$ by $J$. 

\begin{defin} \label{def:61}
Let ${\bf W}$ be an irreducible variety in $\C^{2\ell}$
and  ${\bf I}({\bf W})$ a  homogeneous ideal. 
We say that a  series  $ V= 1 +h.o.t. \in {\mc V}$,
such that  ${\mc A}(V)\in \ker {\mc A}_1$, is  \emph{a  solution to \eqref{eq:AV} on ${\bf W}$} if  
  \be \label{AoW}
  {\mc A}(V)/{\bf I}({\bf W})=0. 
  \ee
\end{defin} 

 We use the notation  $$
{\mc A}(V)|_{\bf W}:=  {\mc A}(V)/{\bf I}({\bf W}).
$$
Note, that for $V\in \mc V$ it can be difficult  to compute $V/{\bf I}({\bf W})$. But if $V$ is a convergent series and ${\bf W}$ is an irreducible 
variety, then $V|_{\bf W }=0 $
yields  $V/{\bf I}({\bf W })=0$.  

\begin{pro} Let ${\bf W}$ be an irreducible variety in $\C^{2\ell}$ 
and    $V= 1 +h.o.t. \in {\mc V}$,
be a  solution to \eqref{eq:AV} on ${\bf W}$.
Then  for any  $(a,b)\in {\bf  W}$ the corresponding system 
\eqref{gs} admits an analytic  first integral  in a neighborhood of the origin. 
\end{pro}
\begin{proof}
    Let 
     $$
{\mc A}( V)=\sum_{k=1}^\infty g_{qk,pk}(a,b), \quad {\rm where } \quad   g_{kq,kp}(a,b)\in R_{qk,pk}.
$$
By \eqref{AoW} 
\be \label{gpw}
 g_{kq,kp}(a,b)|_{{\bf W}}\equiv 0.
\ee
Let $\Psi(x,y)= \mathfrak{T}^{-1}(V)$. By Theorem \ref{th2}
$$
D(\Psi)= \sum_{k=1}^\infty g_{qk,pk}(a,b) x^{qk} y^{pk}.
$$
In view of \eqref{gpw} for any $(a,b)\in {\bf  W}$ the series $\Psi(x,y) $ is a first integral of the corresponding system \eqref{gs}. But then the system also admits an analytic first integral in a neighborhood of the origin \cite{RS_BMS}. 
\end{proof}


Assume that the  variety ${\bf W}$ is an invariant set of the  homogeneous dual  system \eqref{Vvf_sh},
 it has dimension $m$,  and 
admits the   parametrization
\be \label{Wpar}
a_{ij}= f_{ij}(t_1, \dots, t_m), \quad b_{ji}= g_{ij}(t_1,\dots, t_m).
\ee
Let $t=  (t_1,\dots, t_m)$ and 
$$
{\mc H}: \C^m  \to \C^{2 \ell}
$$
be the map defined by \eqref{Wpar}.  We  write \eqref{Wpar} in a short form as 
$$
a=f(t), \quad b=g(t).
$$

Then 
$$
\frac{d  {\mc H}}{d \tau }=J(t) \frac{dt}{d\tau},
$$
where 
$
J(t)=
\frac{\partial { (f, g)}}{\partial t } 
$
is the Jacobi matrix of the map ${\mc H}$.
Since both $ (\dot a, \dot b):= (\dot a_{ij},  \dot b_{ji})\in \C^\ell \times  \C^\ell$, $(i,j)\in S$,  defined by \eqref{Vvf_sh} and $J(t) \frac{dt}{d\tau}$ are vectors from the tangent space of ${\bf W}$, which has dimension $m$, and the rank of $J(t)$ in the generic points of the variety ${\bf W}$  equals $m$,  we can solve the equation 
\be \label{abt}
 (\dot a, \dot b)^\top 
 =J(t) \frac{dt}{d\tau},
\ee
for $ \frac{dt}{d\tau}$ obtaining the vector field 
$$\frac{dt}{d\tau}:= (\dot t_1,\dots, \dot t_m)^\top.$$  
Note that the solution of \eqref{abt} is unique as $J(t)$ has the  full column rank in the generic points of $\mathbf W$.
			
			We define the differential operator
{
					$$
			\mc A|_{\bf W}:\C[[t_1, \dots, t_m]] \mapsto  \C[[t_1, \dots, t_m]].
			$$ 
			by  
			\be \label{Awt}
			\mc A|_{\bf W}(\wt V)= \sum_{i=1}^m  \frac{\partial \wt V}{\partial t_i} \dot t_i+|t| \wt V ,
			\ee
			where $ |t|  $ is $\pqs$ evaluated on $
			\bf W$ using parametrization \eqref{Wpar}.			
			\begin{pro}\label{pro:restriction}
				Operator \eqref{Awt}
				is  a restriction of $\mc A$ to 
				$\bf W$ in the sense
				$$
				{\mc A}(V)|_{\bf W}= \mc A|_{\bf W}(\wt V),
				$$
				where 
				$$
				\wt V(t)=V(f(t),g(t)).
				$$
			\end{pro}
			\begin{proof}
				Let \eqref{Wpar} be a parametrization of $\bf W$. We can write operator \eqref{Aop} as 
	$$
		{\mc A}(V)= ({\rm grad\ } V)\cdot (\dot a,\dot b)^\top+ (\pqs) V. 
	$$
				Then,  in view of \eqref{abt} and the chain rule ${\rm grad\ }\wt V=({\rm grad\ } V)|_{\bf W} J(t)$,   we have 
	$$
		\begin{aligned}
			{\mc A}(V)|_{\bf W}= & \left( ({\rm grad \ } V)\cdot (\dot a,\dot b)^\top\right)|_{{\bf W}} + \left( (\pqs)  V \right) |_{{\bf W}}
					\\
					= &  ({\rm grad\ } V)|_{{\bf W}} \cdot  J(t) \frac{dt}{d\tau}   + \left( (\pqs) V \right) |_{{\bf W}}\\
					=&  ({\rm grad \ }\wt V)  \cdot  \frac{dt}{d\tau}  + (\pqs)|_{{\bf W}}\wt V \\
					= & {\mc A}|_{\bf W}(\wt V).
				\end{aligned}
				$$			
	\end{proof}
		}

\begin{teo}\label{th:61}
Let $\bf W$ be a component of the integrability variety $\vv (B)$ of system \eqref{gs1}.
If $ V=1+h.o.t. $ is a 
solution to \eqref{eq:AV} on $\bf  W$, then $\wt V= V|_{\bf W} $ is a solution to 
\be\label{e:AonW}
\mc A|_{\bf W}(V)=0. 
\ee
\end{teo}
\begin{proof}
Let $\bf W$ be parametrized by \eqref{Wpar}.
Since  $V$ is a solution to \eqref{eq:AV} on $\bf  W$, 
$$
{\mc A}(V)|_{\bf W}=0.
$$
Then by Proposition \ref{pro:restriction} 
$$
  {\mc A}|_{\bf W}(\wt V)=  {\mc A}(V)|_{\bf W} =0.
$$
\end{proof}

The following conjecture claims   that the statement of 
Theorem \ref{th:61} is reversible. 
\begin{conj}\label{conj}
If   equation 
\eqref{e:AonW}
has a formal power series  solution $\wt V(t_1,\dots, t_m)=1+h.o.t.$,
then  equation \eqref{eq:AV} has a solution  {$ V$ such that } $ V|_{\bf W}=1+ h.o.t. $
\end{conj}

As we have mentioned in  the Introduction,  the problem of local integrability for system  \eqref{gs} is to determine all systems 
(in family \eqref{gs}) which admit an analytic first integral in a neighborhood of the origin, equivalently, to determine all systems (in the family), having an analytic or formal integral of the form \eqref{Int}.  

As it has been explained above, this is equivalent to finding the variety of the Bautin ideal of 
system \eqref{gs}.
{ That is, 
given a differential operator  \eqref{Aop}, we have to find the variety of the ideal  $B=\la g_{qk,pk},  k \ \in \N \ra$, 
where  $g_{qk,pk},  k \ \in \N $, are polynomials on the right hand side of \eqref{Vprop}.
}
However, we cannot compute the infinite number of polynomials   $g_{qk, pk}$, $k \in \N $.
So, in practice, we compute the ideals $B_1,\dots, B_s$,
$$ 
B_m=\la g_{qk,pk} \ : k =1, \dots, m \ra, \ \ m=1,2,\dots,
$$
until for the chain of the varieties $\vv(B_1)\supset \vv(B_2)\supset
 \vv(B_3)\supset \dots $   we find that   $\vv(B_s)= \vv(B_{s+1})$ 
 for some $s$.  Then we decompose the variety  $\vv(B_s)$. Let  
\be  \label{Bd}
  \vv(B_s)=\cup_{i=1}^d {\bf W}_i
  \ee  be the irreducible decomposition of $\vv(B_s)$. Then we have  to prove that \eqref{Bd} is 
  the irreducible decomposition of the variety of  the Bautin ideal, that is,   $\vv(B)=  \vv(B_s)=\cup_{i=1}^d {\bf W}$.

  If Conjecture \ref{conj} is true then 
  it is sufficient to prove that 
  for {each} $i \in \{1, \dots, d\}$, 
there is a $V_i\in \mc V$, such that 
  $$
  {\mc A}(V_i)|_{{\bf W}_i}=0. 
  $$


One possibility to   solve equation \eqref{e:AonW} is 
using the Darboux method as follows. 

Let 
$\C[[t_1^\pm, \dots, t_m^\pm]]$ be the 
ring  of formal Laurent power series and $\mc B$ a differential operator acting on the ring,
\be \label{AopH} 
{\mc B}:  \C[[t_1^\pm, \dots, t_m^\pm]]
\to \C[[t_1^\pm, \dots , t_m^\pm]].
\ee


We say that 
$f(t,t^-) \in  \C[[t_1^\pm, \dots, t_m^\pm]]  $ is a Darboux factor of operator \eqref{AopH}, 
if there is 
a   $k(t) \in  \C[[t_1^\pm, \dots, t_m^\pm]]  $ such that 
$$
 {{\mc B}}(f)= k f.
$$ 
The series  $k(a,b)$ is called the cofactor of $f(a,b)$. 

The following statement is obvious.
\begin{pro} \label{propDarb}
 Assume that for  operator \eqref{AopH} there are the  Darboux factors  
$f_1, \dots, f_s$  with the corresponding cofactors $k_1, \dots, k_s$ such that 
$
k_1 \alpha_1+\dots +k_s \alpha_m \equiv 0 $ for some $\alpha_1,\dots, \alpha_s \in \C.
$
Then  the function 
$$  V= f_1^{\alpha_1} f_2^{\alpha_2} \cdots f_s^{\alpha_s}$$
satisfies
$$
 {\mc B}(V)=0. 
$$
\end{pro}

\section{The quadratic system}

In this section we treat in some detail the case of 
quadratic system \eqref{vgExfamily}
and discuss the relations between  the  known 
first integrals of system \eqref{vgExfamily},
 the dual series and 
 Conjecture \ref{conj}.

The following operator of the form \eqref{Aop} is associated to system \eqref{vgExfamily}:
\begin{equation}\label{op} 
    \begin{aligned}
        \mathcal{A}(V)&=
  { a_{01}}  \qs  \frac {\partial V} {\partial   { a_{01}}} +   \
{ a_{10}} \ps  \frac {\partial V} {\partial   { a_{10}}} \
+   { a_{-12}} (  - \ps  + 2 \qs ) \frac {\partial V} {\partial   { a_{-12}}} \\&+  { b_{01}} \qs  \frac {\partial V} {\partial   { b_{01}}} +   \
{ b_{10}} \ps  \frac {\partial V} {\partial   { b_{10}}} \
+   { b_{2,-1}} (2  \ps  - \qs ) \frac {\partial V} {\partial   { b_{2,-1}}} \\  &+ ( \pqs )V,     
    \end{aligned}
\end{equation}
    where
\be \label{absum}
\ps = 1 - (a_{10}+ a_{01}+ a_{-12}), \ \ \qs =-1+ b_{01}+ b_{10}+ b_{2,-1}.  
\ee 
The next statement gives an interesting relation
between the components of the integrability variety and the invariant sets of the homogeneous dual system \eqref{Vvf_sh}.

\begin{pro}
   The components of the variety of the Bautin ideal 
   of system \eqref{vgExfamily} are  invariant sets 
   of system \eqref{Vvf_sh}. Moreover, they  are 
   intersections   of invariant algebraic surfaces 
   of system \eqref{Vvf} defined by homogeneous polynomials.
\end{pro}
\begin{proof}
    One possibility to prove that a set $W\subset \C^{2\ell}$ defined by a parametrization  is 
    an invariant set of system \eqref{Vvf_sh} is to
    compute the Jacobian of the parametrization 
    and to check that on  $W$ the vector field of \eqref{Vvf_sh} is in the vector space spanned by the columns of the Jacobian. However, it involves tedious computations. 

Another possibility is provided by the observation that 
a polynomial equation $f(a,b)=0$ defines an invariant algebraic surface of system \eqref{Vvf_sh} if and only if $f$ is a Darboux factor of system \eqref{Vvf_sh}
with a polynomial cofactor. Straightforward computations show that  this   is indeed the case for all polynomials defining the  associate primes $J_1,\dots, J_4$
(presented at the end of Section \ref{BauId})
of the Bautin ideal of system \eqref{vgExfamily}.
\end{proof}

Based on the proposition, we can propose the following conjecture.
\begin{conj}
The components of the variety of the Bautin ideal 
   of system \eqref{gs} are  invariant sets 
   of system \eqref{Vvf_sh} defined as  
   intersections   of invariant algebraic surfaces 
   of system \eqref{gs}, which are  defined by homogeneous polynomials.
\end{conj}


We denote by $\widehat {\mc A}  $ the homogeneous 
part of operator \eqref{op}, that is, the operator $\mc A$ without the last term $(\pqs)V$ and use the similar notation for $\mc A$ restricted to a variety. 

Case 1.    The  variety  $\vv_1$  
   of the  ideal 
   $J_1= \la 2 a_{10} -b_{10}, 2 b_{01}- a_{01}\ra$ is parameterized by   
    $$  
     a_{01} = 2  b_{01}, \ \ 
      b_{10} = 2  a_{10},
    $$ 
and operator \eqref{Awt} takes the form 
\begin{multline*}
\mathcal{A}_{\vv_1}(V):=
  { a_{10}} \ps  \frac {\partial V} {\partial   { a_{10}}} \
+   { a_{-12}} ( -\ps +2 \qs ) \frac {\partial V} {\partial   { a_{-12}}} +   \\
{ b_{01}}   \qs  \frac {\partial V} {\partial   { b_{01}}} +   \
{ b_{2,-1}} (2 \ps -\qs ) \frac {\partial V} {\partial   { b_{2,-1}}} + (\pqs) V,
\end{multline*}
where $\ps = 1-( a_{10}+ 2 b_{01}+ a_{-12}), \qquad \qs =-1+ b_{01}+ 2 a_{10}+ b_{2,-1}. $

A simple computation shows  that 
$$
V=    1 -  a_{-12}/3  -  b_{2,-1}/3 -  a_{10} -  b_{01}$$
is a solution to 
$    { \cal A}_{\vv_1} (V)=0.$ 
Since 
\begin{multline*}
 {\mc A}(V)=- a_{01} - 2  a_{10} + 2  a_{01}  a_{10} + 2  a_{10}^2 + 2  a_{10}  a_{-12} + 2  b_{01} + \\ a_{01}  b_{01} - 
 2  b_{01}^2 +  b_{10} -  a_{10}  b_{10} -  a_{-12}  b_{10} - 2  b_{01}  b_{10} +  a_{01}  b_{2,-1} - 2  b_{01}  b_{2,-1},
\end{multline*}
we have that  
 ${ \cal A}(V)\equiv 0\mod {\bf I}(\vv_1)$, 
 in agreement with Conjecture \ref{conj}.

System \eqref{vgExfamily} corresponding to this case  is Hamiltonian system. 
We treat the general family of Hamiltonian  systems \eqref{gs} in Section \ref{s:HamSys}.


Case 2.  The corresponding system  \eqref{vgExfamily} is 
\be\label{vgC2}
\begin{aligned}
\dot x &= \phantom{-(} 
          x - a_{10} x^2  - a_{-12} y^2  \\
\dot y &=-(y-b_{2,-1} x^2  - b_{01} y^2)
\end{aligned}
\ee
and the corresponding equation \eqref{e:AonW}  is 
written for the operator 
\begin{multline}\label{op_c2} 
\mathcal{A}|_{ {\bf V}_3}(V):=
{ a_{10}} \ps  \frac {\partial V} {\partial   { a_{10}}} \
+   { a_{-12}} ( -\ps +2 \qs ) \frac {\partial V} {\partial   { a_{-12}}} + \\  { b_{01}}  \qs  \frac {\partial V} {\partial   { b_{01}}} +    { b_{2,-1}} (2\ps -\qs) \frac {\partial V} {\partial   { b_{2,-1}}}  +(\pqs) V,
    \end{multline}
    where
$
\ps =1-( a_{10}+ a_{-12}), \ \qs  =-1+ b_{01}+  b_{2,-1}  
$
(we can also  treat   operator \eqref{op_c2} as operator 
\eqref{Aop} associated to the set $S=\{(1,0), (-1,2) \}$).

It is not difficult to see that 
the line
$$
\ell= 1+ c_{10} x + c_{01} y 
$$
is an invariant line of system \eqref{vgC2} if 
 $  c_{10}$ is a root of the cubic equation 
\be \label{eq:cubc2}
 a_{10}  b_{01}  b_{2,-1} -  a_{-12}  b_{2,-1}^2 + ( a_{10}^2 +  b_{01}  b_{2,-1})  c_{10} + 2  a_{10}  c_{10}^2 +  c_{10}^3=0.
\ee

Thus, 
in the generic case when all three roots $ c_{10}^{(1)},   c_{10}^{(2)},$ and $ c_{10}^{(3)}$ 
of \eqref{eq:cubc2} are different,  system \eqref{vgC2} has the Darboux first  integral 
$$
\Phi(x,y)= (1+ c_{10}^{(1)} x + c_{01}^{(1)} y)^{\alpha_1}  (1+ c_{10}^{(2)} x + c_{01}^{(2)} y)^{\alpha_2}  (1+ c_{10}^{(3)} x + c_{01}^{(3)} y)^{\alpha_3}, 
$$
where $c_{01}^{(i)}$ is a solution of
$$ (c_{01}^{(i)})^2+  b_{01}  c_{01}^{(i)}-  a_{-12}c_{10}^{(i)}=0
$$
and 
$\alpha_1, \alpha_2, \alpha_3$ satisfy the equation 
$$
\sum_{i=1}^3 \alpha_i ( c_{10}^{(i)}  x -   ( a_{10}  c_{10}^{(i)} + \frac{( c_{10}^{(i)})^2}{  b_{2,-1} } ) y\equiv 0. 
$$
Then 
\be \label{int_c2}
\Psi(x,y)=\frac{\Phi(x,y)}{ \alpha_1 ( c_{10}^{(1)})^3+  \alpha_2 ( c_{10}^{(2)})^3+ \alpha_3 ( c_{10}^{(3)})^3 }
\ee
is an analytic  first integral of the form \eqref{e:Psi} (see \cite[Section 3.7]{RS}).
However,  $\Psi(1,1)$ is not a formal power series, so map \eqref{eq:T} is not defined on 
$\Psi(x,y)$.

By Theorem  \ref{propAV} equation \eqref{op_c2} has a power series solution, but  we cannot find it using the Darboux first integral \eqref{int_c2} of system \eqref{vgC2}. 

\begin{prob} {\rm 
How to prove the existence of power series solution to ${\mc A}|_{\vv_3 } (V)=0$ without referring to 
Theorem  \ref{propAV}?}
\end{prob}

Case 3. 
The variety ${\bf V}_3$  is the intersection of two hyperplanes and a five-dimensional torus, so
we parametrize it as 
\be \label{par3tor}
 a_{10} = -2 t y_1,\
 a_{01} = y_2t^{-1},\
 a_{-12} = y_1 t^{-1},\
 b_{2,-1} = y_2 t,\
 b_{10} = y_1 t, \
 b_{01} = -2 y_2 t^{-1}.
\ee 
The associated  degree map on $\C[t, y_1, y_2]$ is
$$
L^{(3)}(\nu)= {1 \ch -1} \nu_1+ {0 \ch 1}\nu_2+{1 \ch 0} \nu_3. 
$$
Using  parametrization \eqref{par3tor} 
operator \eqref{op} on the variety ${\bf V}_3$ is written as 
\be \label{Av3def}
\mathcal{A}_{\vv _3}(V):=
  \frac{\partial V}{\partial  { y_1}} \dot y_1 + \frac{\partial V}{\partial  { y_2}} \dot y_2 + \frac{\partial V}{\partial  { t}} \dot t   -  V (  y_1 t^{-1} - 3 t y_1 + 3 y_2t^{-1} - t y_2 ),
\ee
where
$$
\begin{aligned}
\dot y_1=& y_1 (-1  + t y_1 - 2  y_2t^{-1} + t  y_2),\\ 
\dot y_2= &  y_2 (1 - y_1 t^{-1} + 2 t y_1  - y_2 t^{-1}, \\ \dot t\phantom{_2} =&  t( 2  - y_1 t^{-1}  + t y_1 + y_2 - t y_2).
\end{aligned}
$$

Direct computations show that 
the Laurent polynomials 
$$
f_1=1 + 2 t y_1 - y_1^2 + 2 y_2 t^{-1} + 2 y_1 y_2 - y_2^2
$$
and 
$$
f_2=-\frac{3 t^2  {y_1}^3}{2  {y_2}}+\frac{3 t^2
    {y_1}^2}{2}-\frac{3  {y_2}^3}{2 t^2
    {y_1}}+\frac{3  {y_2}^2}{2 t^2}+\frac{t
    {y_1}^4}{2  {y_2}}-\frac{3 t
    {y_1}^3}{2}+\frac{3}{2} t  {y_1}^2
    {y_2}-\frac{ {y_1}^2  {y_2}}{2
   t}+
   $$
   $$
   \frac{ {y_2}^4}{2 t  {y_1}}-\frac{1}{2} t
    {y_1}  {y_2}^2+\frac{3  {y_1}
    {y_2}^2}{2 t}+3 t  {y_1}-\frac{3
    {y_2}^3}{2 t}+\frac{3  {y_2}}{t}-\frac{3
    {y_1}^2}{2}+3  {y_1}  {y_2}-\frac{3
    {y_2}^2}{2}+1
$$
are the Darboux factors of the operator 
$
\widehat {\mathcal{A}}_{\vv _3}(V) 
$
with the cofactors 
$$    
k_1= \frac{2 \left(t^2  {y_1}- {y_2}\right)}{t} \quad  {\rm and} \quad  
k_2= \frac{3 \left(t^2  {y_1}- {y_2}\right)}{t},
$$ 
respectively.
 Thus, by Proposition \ref{propDarb},
    $$
 \wt  U=f_1^{-3}f_2^2
$$
is a solution to 
$
\hat{ \mathcal{A}}(\wt U)=0.
$

The series expansion of $U$ is 
$$
\wt U=1 - 6 y_1 y_2+h.o.t.
$$
Then 
\be \label{Vc3}
\wt V=\frac{U-1}{ - 6 y_1 y_2 } 
\ee
is a   solution to 
\be \label{Av3eq}
{ \mathcal{A}_{\vv_3}}(V)=  0
\ee
(
to  obtain the solution \eqref{Vc3} we have divided $U-1$ by the terms of $U$ such that $L^{(3)}(\mu,\nu)=(1,1)$
).

Solution $V(y_1,y_2,t)$  given by  \eqref{Vc3} is a solution to \eqref{Av3eq}, where the differential operator is defined by \eqref{Av3def}. We now show that there is a series $V$ which is solution 
to 
$$
{\mc A}(V)=0, 
$$
where $\mc A$ is defined by \eqref{op}, on $\vv_3$ in the sense of Definition  \ref{def:61}.


Note that  parametrization \eqref{par3tor} is associated  in  the  natural way with the rational map $\pi$ acting by 
$$   
 a_{10} \mapsto  -2 t y_1,\
 a_{01} \mapsto y_2 t^{-1}, \ 
 a_{-12} \mapsto  y_1  t^{-1}, \
 b_{2,-1} \mapsto  y_2 t, \
 b_{10} \mapsto y_1 t, \
 b_{01} \mapsto   -2 y_2t^{-1}.
$$

To 
compute the preimages  of $f_1$ and $f_2$ under $\pi$ we can  proceed similarly to the rational implicitization (see  \cite[Chapter 3]{Cox}).  That is, to compute $\pi^{-1}(f_1)$ we eliminate from  the ideal
\begin{multline*}
\la 1 - w t y_1 y_2  ,  b_{10} - y_1 t, \
 a_{01} - y_2 t^{-1},\
 a_{-12} - y_1 t^{-1},\
 b_{2,-1} - y_2 t,\
 b_{01} +2 y_2 t^{-1},\
 a_{10} +2 t y_1, f_1
  \ra,
\end{multline*}
the variables $w, t, y_1 $ and $ y_2 $
obtaining  
$$
F_1=\pi^{-1}(f_1)=
1 + 2  a_{01} + 2  b_{10} -  a_{-12}  b_{10} -  a_{01}  b_{2,-1} + 2  a_{-12}  b_{2,-1}.
$$
Computing the fourth elimination ideal of the ideal 
$$
\la 1 - w t y_1 y_2  ,  b_{10} - y_1 t, \
 a_{01} - y_2 t^{-1},\
 a_{-12} - y_1 t^{-1},\
 b_{2,-1} - y_2 t,\
 b_{01} +2 y_2 t^{-1},\
 a_{10} +2 t y_1, f_1\ra 
$$
we find that it contains the polynomial 

$
  F'_2=-3  a_{-12}  b_{10}^3 +  a_{-12}^2  b_{10}^3 - 3  a_{01}^3  b_{2,-1} + 2  a_{-12}  b_{2,-1} + 
 6  a_{01}  a_{-12}  b_{2,-1} + 3  a_{01}^2  a_{-12}  b_{2,-1} + 6  a_{-12}  b_{10}  b_{2,-1} - 
 3  a_{-12}^2  b_{10}  b_{2,-1} + 3  a_{-12}  b_{10}^2  b_{2,-1} - 3  a_{-12}^2  b_{10}^2  b_{2,-1} + 
  a_{01}^3  b_{2,-1}^2 - 3  a_{01}  a_{-12}  b_{2,-1}^2 - 3  a_{01}^2  a_{-12}  b_{2,-1}^2 + 
 6  a_{-12}^2  b_{2,-1}^2 + 3  a_{01}  a_{-12}^2  b_{2,-1}^2 -  a_{-12}^3  b_{2,-1}^2 + 
 3  a_{-12}^2  b_{10}  b_{2,-1}^2 -  a_{-12}^2  b_{2,-1}^3.
$

Let
$$
F_2=\frac{ F'_2}{2  a_{-12}   b_{2,-1} }.
$$
Computations show that 
for 
$K= 
- a_{01} +  b_{10}
$
$$
(\widehat{\mc A}(F_1)- 2 K F_1)/{\bf I}(\vv_3)=0, \quad    (\widehat{\mc A}(F_1)- 3 K F_2)/{\bf I}(\vv_3)=0,
$$ 
so $ 2 K$ and 
$ 3 K  $ are the cofactors of $F_1$ and $F_2$, respectively,  on $\vv_3.$
Then 
$$
 U= F_1^{-3} F^2_2
$$
is a Darboux solution to 
$$
\widehat{ 
{\mc A}}_{\vv_3}(V)=0
$$
and
$$
 V=
 \frac{ U-1}{ - 6  a_{01}  b_{10}}
$$
is a solution to \eqref{op} on $\vv_3$.
This supports Conjecture \ref{conj}. 

\begin{remark}
{\rm 

An obvious  rational parametrization of $\vv_3$ is 
\be \label{par_c3}
a_{-12} = \frac{ a_{01}  b_{10}}{b_{2,-1}}, \quad b_{01} = -2  a_{01}, \quad a_{10} = -2  b_{10}.
\ee
For this parametrization 
operator \eqref{Awt}   takes the form 
\begin{multline}
\mathcal{A}^{(1)}_{\vv _3}(V):=
 { a_{01}} ( \qs  \frac{\partial V}{\partial 
 { a_{01}}}+ { b_{10}} \ps  \frac{\partial V}{
\partial  { b_{10}}}+\\ { b_{2,-1}} (2 \ps -\qs ) \frac{\partial V}{\partial  { b_{2,-1}}} +(\pqs)  V,
\end{multline}
where $\ps $ and $\qs $ are expressions 
\eqref{absum} evaluated using \eqref{par_c3}.
It is possible to find a Darboux solution to the equation
$$
\mathcal{A}^{(1)}_{\vv _3}(V)=0,
$$
however,  the computations are  more tedious compared to the ones for 
parametrization \eqref{par3tor}.

}
\end{remark}

Case 4. The ideal $J_4$ is the Sibirsky ideal of system \eqref{vgExfamily}, so by Proposition \ref{pro:6} for each 
$(a,b)\in \vv (J_4)$ the corresponding system \eqref{vgExfamily}
has a first integral of the form \eqref{Psi}.

Using matrix \eqref{Mfr} (equivalently, \eqref{par_tor}) we immediately 
have the   parametrization 
$$
 a_{10} = y_1  t, \
 a_{01} = y_2 t^{-1} , \
 a_{-12} = y_3 t^{-3}, \
 b_{2,-1} = y_3 t^3 
 b_{10} = y_2 t, \
 b_{01} = y_1t^{-1} 
$$
of $\vv_4$, 
which we write in the short form as
\be \label{abf}
(a,b)=f(t,y_1,y_2,y_4).
\ee

The degree map 
$$
L^{(4)}(\nu)={\frac12 \ch -\frac{1}2 }\nu_1+{\frac12 \ch \frac12 }\nu_2 +{\frac12 \ch \frac12 }\nu_3+{\frac12 \ch \frac12 }\nu_4
$$
defines the multigraded  ring $\C[t,t^{-1},  y_1, y_2, y_3]$.

After rescaling by 
$$
\frac{\left(t^2-1\right) \left(t^4
    {y_3}+t^2
   (- {y_1}+ {y_2}+ {y_3})+ {y_3}\right)}{2 t^3}
$$
 the differential equation
 on the torus can be written as  
\be \label{c4t}
\widehat{\mathcal{A}}_{\vv _4}(V):=
  \frac{\partial V}{\partial  { y_1}}  y_1 + \frac{\partial V}{\partial  { y_2}}  y_2 + 
   \frac{\partial V}{\partial  { y_3}}  y_3 +  
  \frac{\partial V}{\partial  { t}} \dot t   =- 2  V ,
\ee
where 
$$
\dot t=-\frac{t \left(t^2 \left(t \left(t^3
    {y_3}+t
   ( {y_1}+ {y_2})-2\right)+ {y_1}
   + {y_2}\right)+ {y_3}
   \right)}{\left(t^2-1\right) \left(t^4  {y_3}+t^2
   (- {y_1}+ {y_2}+ {y_3})+ {y_3}\right)}.
$$

\begin{pro}\label{pro:c4}  
   Equation \eqref{c4t} 
   has a solution $\wt V=1+h.o.t. \in \wt {\mc V}$.
\end{pro}
\begin{proof}
   By Proposition \eqref{pro:6} for each $(a,b)\in \vv_4$ the corresponding system  \eqref{vgExfamily}
   has  local analytic  first integral. 
   Therefore, there is a $V=1+h.o.t. \in \mc V$ which is a solution to \eqref{op} on $\vv_4$. Then, by Theorem \ref{th:61} $\wt V= V(f)  $, where $f$
   is defined by \eqref{abf} is a solution to \eqref{c4t}.
\end{proof}

\begin{prob}
   {\rm  Can one prove Proposition \ref{pro:c4} without refereeing to Theorem \ref{th:61}?
Is it possible to find a solution to \eqref{c4t}?
}
\end{prob}

 

\section{Hamiltonian systems} \label{s:HamSys}

In this section, we treat the general case of $1:-1$ resonant  Hamiltonian polynomial  systems. 
It is easy to see that the component of Hamiltonian systems in the space of parameters 
of system \eqref{gs} with $p=q=1$ is the variety of the ideal  
$$
\begin{aligned}
I_{H}
\stackrel{\textrm{def}}{=}\,
\la \, (i+1) \, & a_{ij} - (j+1) \, b_{ij}
: 
i \ge 0 \text{ and } (i,j) \in S
        \text{ and } (j,i) \in S \, \ra \\
&\cap\,
\la \, a_{ij}, b_{ji}
:
i \ge 0 \text{ and } (i,j) \in S
        \text{ and } (j,i) \not\in S \, \ra  \,.
\end{aligned}
$$
Let $S_H$ be a subset of $S$,
$$
S_H=\{(i,j)\in S\ : \ i\ge j \ {\rm or\ }\ i= -1 \}.
$$

\begin{teo}\label{th3}
Let the variety of $I_H$ be parametrized as follows:
\be
\label{varH}
\begin{aligned}
b_{ij}= & \frac{i+1}{j+1} a_{ij}, \ a_{ji}=\frac{i+1}{j+1} b_{ji}, \ (i,j)\in S_H, i\ne -1 \\ 
\ a_{ij} = & a_{ij},\ b_{ji}=b_{ji},  \ (i,j) \in S_H, \ i= -1.
\end{aligned}
\ee
Then the equation 
\be
\label{eq:A_H}
 {\mc A}|_{\vv(I_H)}(V)=0, 
\ee
where $\mc A$ is defined by \eqref{Aop},
has the polynomial solution
$$
V_H=1-\sum_{(i,j)\in S_H} (\frac{1}{j+1}a_{ij}+\frac{1}{j+1}b_{ji}  ).
$$ 
\end{teo}
\begin{proof}  We can write the operator ${\mc A}|_{\vv(I_H)}(V)$ as the sum 
$$
{\mc A}|_{\vv(I_H)}(V)=A_L(V)+A_Q(V)+(\pqs) V,
$$
where 
$$
A_L=\sH \left[ (i-j) a_{ij}\frac{\partial}{\partial a_{ij}}+ 
(j-i) b_{ji}\frac{\partial}{\partial b_{ji}}\right],
$$
$$
A_Q=\sH \left[ ( i \as +j \bs) a_{ij} \frac{\partial}{\partial a_{ij}}+ 
 (i \bs +j \as) b_{ji} \frac{\partial}{\partial b_{ji}}  \right]. 
$$
Clearly, 
$$
{\mc A}|_{\vv(I_H)}(V_H)=A_L(V_H)+A_Q(V_H)+ (\pqs)-(V_H-1) (\pqs).
$$

Observe that in view of \eqref{varH}
$$
\as=-\sH \left( a_{ij}+\frac{i+1}{j+1} b_{ji}\right),\quad  \bs=\sH \left( b_{ji}+\frac{i+1}{j+1} a_{ij}\right), 
$$
so,
$$
\pqs = -\sH \left( \frac{i-j}{j+1} a_{ij}+ \frac{j-i}{j+1} b_{ji}  \right).
$$ 
Then 
$$
A_L(V_H)+\pqs = -\sH \left(\frac{i-j}{j+1} a_{ij}+ \frac{j-i}{j+1} b_{ji} \right) +\pqs =0.
$$
Computations give
$$
A_Q(V_H)-(V_H-1) (\pqs)=$$
$$  -\sH \left[  i \as +j \bs ) \frac{a_{ij}}{j+1 } +
i \bs + j \as) \frac{b_{ji}}{j+1 } \right] -
$$
$$
(V_H-1) (\pqs)=$$
$$ 
\sH \left[   -\as \left(\frac{i+1}{j+1} a_{ij} +b_{ji} \right)- \bs \left(\frac{i+1}{j+1} b_{ji} +a_{ij} \right) \right] =
$$
$$
-\as  \bs- \bs \as=0.
$$
Thus, \eqref{eq:A_H} holds. 
\end{proof}

\begin{remark}
{\rm
We have proved Theorem \ref{th3} directly using properties of operator  ${\mc A}|_{\vv(I_H)}$.
A simpler proof is as follows.  
Observe that the function 
$$
H=- x y+\sum_{(i,j)\in S_H} (\frac{1}{j+1}a_{ij }x^i y^{j+1}  + \frac{1}{j+1}b_{ji} x^{j+1} y^i)
$$
is in the Hamiltonian of  system \eqref{gs} with the parameters given by \eqref{varH}, that is, $\mathcal{D}(H)=0.$
Therefore, by Theorem \ref{th:61} we conclude that \eqref{eq:A_H} holds.}
\end{remark} 

To conclude, 
in this paper  we investigated and described  a duality  between analytic or formal first integrals of certain  polynomial  systems of ODEs  and power series in the parameter space of the system which are solutions 
to  first order linear PDEs. This duality provides new insights into
the structure of integrable systems and their first integrals. Using this duality we were 
able to find solutions to some PDEs (or at least to prove their existence), which, 
we believe, is not possible to find using known methods from the theory of linear PDEs and differential algebra. We believe our  findings could offer valuable insights and tools for future research in these fields.
		
		\section*{Acknowledgments}
		The first author is supported by the Slovenian Research and Innovation  Agency (core research program P1-0288) and the second author is supported by the Slovenian Research  and Innovation  Agency (core research program P1-0306).
		The authors also acknowledge the support by the project  101183111-DSYREKI-HORIZON-MSCA-2023-SE-01 “Dynamical Systems and Reaction Kinetics Networks”.



\end{document}